\documentclass{article}[11pt]
\usepackage{amssymb, amsthm, amsmath}
\usepackage[a4paper]{geometry}
\usepackage{amssymb,latexsym,amscd}
\usepackage{epsfig}
\usepackage[active]{srcltx}
\usepackage{psfrag}
\usepackage{graphics,graphicx}
\usepackage{pstricks,pst-node,pst-tree}
\usepackage{braket}
\usepackage{color}
\usepackage{natbib}
\usepackage{bm}

\definecolor{r}{rgb}{1,0,0}
\definecolor{b}{rgb}{0,0,1}
\definecolor{g}{rgb}{0,1,0}

 \newcommand{\tr}[1]{}
 \newcommand{\tb}[1]{#1}

\newcommand{\fra}[2]{\textstyle{\frac{#1}{#2}}}
\newcommand{\beqn}{\begin{eqnarray}\begin{aligned}}
\newcommand{\eqn}{\end{aligned}\end{eqnarray}}
\newcommand{\lr}[1]{{\langle #1 \rangle}}
\newcommand{\lrr}[1]{{\langle #1 \rangle}_{\mathbb{R}}}
\newcommand{\lrc}[1]{{\langle #1 \rangle}_{\mathbb{C}}}
\newcommand{\s}[1]{\{ #1 \}}
\newcommand{\lie}[2]{\left[#1,#2\right]}
\newcommand{\id}{\mathbf{1}}

\newcommand{\CC}{\mathbb{C}}
\newcommand{\LL}{\mathfrak{L}}
\newcommand{\SG}{\mathfrak{S}}

\theoremstyle{plain}
\newtheorem{thm}{Theorem}[section]
\newtheorem{prop}{Proposition}[section]

\newtheorem{res}{Result}
\newtheorem{rem}{Remark}

\newtheorem{lem}[thm]{Lemma}

\theoremstyle{definition}
\newtheorem{mydef}[thm]{Definition}

\newif\ifprivate
\def\xbar{\vskip0.09in\hrule\vskip0.06in}
\def\private#1{\ifprivate \xbar {\em #1} \xbar
\else \fi}

\def\???{\ifprivate {\bf {???}} \marginpar{{\Huge {\bf ?}}}
\else \fi}

\title{Lie Markov Models}
\author{Jeremy Sumner$^{1,*}$, Jes\'us Fern\'andez-S\'anchez$^{2}$ and Peter Jarvis$^{1,\dagger}$}

\begin{document}
\maketitle

\begin{abstract}
Recent work has discussed the importance of multiplicative closure for the Markov models used in phylogenetics. 
For continuous-time Markov chains, a sufficient condition for multiplicative closure of a model class is ensured by demanding that the set of rate-matrices belonging to the model class form a Lie algebra. 
It is the case that some well-known Markov models \emph{do} form Lie algebras and we refer to such models as ``Lie Markov models''.
However it is also the case that some other well-known Markov models unequivocally \emph{do not} form Lie algebras (GTR being the most conspicuous example).

In this paper, we will discuss how to generate Lie Markov models by demanding that the models have certain symmetries under nucleotide permutations.
We show that the Lie Markov models include, and hence provide a unifying concept for, ``group-based'' and ``equivariant'' models. 
For each of two, three and four character states, the full list of Lie Markov models with maximal symmetry is presented and shown to include interesting examples that are neither group-based nor equivariant. 
We also argue that our scheme is pleasing in the context of applied phylogenetics, as, for a given symmetry of nucleotide substitution, it provides a natural hierarchy of models with increasing number of parameters. 
\end{abstract}

\vfill
\hrule\mbox{}\\
\thanks{\footnotesize{
\noindent
$^1$School of Mathematics and Physics, University of Tasmania, Australia\\
$^2$Departament de Matem\`atica Aplicada, Universitat Polit\`ecnica de Catalunya, Spain\\
$^{\dagger}$Alexander von Humboldt Fellow, $^{\ast}$ARC Research Fellow\\
\textit{keywords:} phylogenetics, Lie algebras, Lie groups, representation theory, symmetry, Markov chains\\
\textit{Corresponding author:} Jeremy Sumner, jsumner@utas.edu.au
}
}

\newpage

\section{Introduction}

Continuous-time Markov chains are fundamental to the implementation of, and philosophy behind, many phylogenetic methods.
Likelihood and Bayesian phylogenetic methods usually proceed by attempting to fit a single ``rate-matrix'' globally across a proposed evolutionary tree history (see, for example, Chapters~2 and 3 of \citet{gascuel2005}).
These rate-matrices are chosen from some restricted class or ``model'' that is defined by a certain set of constraints on the elements of a generic rate-matrix.
These constraints define a set of free parameters that usually correspond to unknown evolutionary quantities such as base composition, mutation rates and the timing of speciation events (these last two are often by necessity confounded together simply as ``edge lengths'').
Even in phylogenetic distance methods, it is usually the case that the theoretical justification of a given distance estimator is taken from a continuous-time Markov model (for example the general Markov model for the ``log-det'' \citep{steel1994} distance or the HKY distance taken from its corresponding model \citep{felsenstein2004}).

A \emph{homogeneous} Markov chain satisfies the condition that the probability transition rates are constant in time.
In the phylogenetic context this means that the rates are unchanged throughout evolutionary history.
Of course, this is used as an approximation to biological reality where it is well documented that transition rates are not only time-dependent \citep{ho2005,ho2007}, but also vary across the different lineages of the evolutionary tree \citep{lockhart1998}.
Methods to cope with these issues have been explored by various authors: \citet{tuffley1997a} proposed the ``covarion'' model where a switching process allows sites to alternate between ``on'' and ``off'' states. 
\citet{drummond2006}, proposed a method that introduces an overall scaling factor for the transition rates that is sampled randomly (at branching events, for example), and the methods presented in \citet{whelan2008} are more general still with a switching process that allows for alteration of individual rates.
The simulation package discussed in \citet{fletcher2009} provides further evidence that these issues are of ongoing importance to phylogenetic analysis.

Our philosophy is to remain agnostic as to whether evolutionary rates have changed in the past or, indeed, whether it is possible to statistically detect this change via analysis of present day molecular data.
We follow an approach that allows for the biological possibility that there is likely to have been a smooth (or even abrupt) change of each \emph{individual}  transition rate independently occurring across the evolutionary tree (and not necessarily restricted to branching events).
This discussion leads naturally to confronting the possibility (at least theoretically) that the phylogenetic process is not homogeneous and is more accurately modelled as an \emph{inhomogeneous} continuous-time Markov chain; where the rate-matrix is far from constant and ultimately is allowed to vary, smoothly or otherwise, as a function of edge length parameters of the evolutionary tree.

Of course, given the bias/variance tradeoff of statistical analysis \citep{burnham2002}, modelling phylogenetic evolution as a inhomogeneous process is statistically implausible in practice (we would effectively be replacing a small number of parameters by an infinite continuum).
Indeed this is where the methods discussed in \citet{drummond2006}, where rates may change but only at branching events, can be seen as somewhat of an intelligent compromise between a (statistically tractable) homogeneous model and a (biologically realistic) inhomogeneous model.
Another approach would be to abandon the continuous-time hypothesis and work with discrete Markov chains (or equivalently ``algebraic'' models \citep{pachter2005}).
However this approach introduces many free parameters and suffers from a lack of interpretation, as it is unclear what the free algebraic parameters mean in biological terms (such as divergence times and molecular rates), except with reference to the corresponding continuous-time approach.

\private{OLD Jesus:I would remove this remark on algebraic models. The algebraic parameters are the probabilities of nucleotide substitution, so they make sense independently of any reference to time or taxes.  On the other hand, the only difference between ``algebraic'' models and the usual evolutionary models is what one considers to be the parameters of the model.}
\private{Jeremy: I have retained (but modified) the remark as I think it is exactly because algebraic models are defined independently of reference to time (or molecular rates) that means that algebraic models are difficult to interpret in biological terms (i.e. divergence times and molecular rates!) Sure, the algebraic models give probabilities of nucleotide substitution, but probabilities that are conditional on what???
I realize that you might be uncomfortable about this comment, as it may seem a little controversial. My feeling is that it doesn't hurt to stir up a little controversy!}
\private{Jesus: A little controversy is okay. Let me think about it. }
\private{Jeremy: What is your opinion about this controversy?}
An available resolution of these issues is to observe that it is possible to continue to model phylogenetic processes as being homogeneous, but interpret the transition rates that are fitted globally across the tree (or at least non-locally) as a kind of ``average'' of the true inhomogeneous process.
It is this perspective that we take in this work and it leads directly to the concept of \emph{multiplicative closure} for continuous-time Markov chains.
It will be shown that models that are multiplicatively closed have the property that, even in their inhomogeneous formulation, it is possible to interpret their average behaviour as a homogeneous process.
It is then the purpose of this article to discuss sufficient conditions for multiplicative closure of continuous-time Markov chains.
In order to generate particular examples of closed models, we exploit symmetry properties of DNA substitution rates to present a scheme that creates a hierarchy of closed Markov models based on the number of free parameters available.

In \S\ref{sec2} we give basic definitions of multiplicative closure and Lie Markov models. 
To achieve this we review the required Lie theory, and we discuss the Lie algebra of the general Markov model.
As an example to motivate the general procedure, in \S\ref{sec3} we specialize to the case of binary Markov chains and give a complete description of Lie Markov models in this case.
In \S\ref{sec4} we discuss the symmetry properties of Markov models and explain how symmetry can be used to assist in the search for Lie Markov models.
Here we also prove that \emph{equivariant} (see \cite{draisma2008,casanellas2010}) and \emph{group-based} models are examples of Lie Markov models.
In \S\ref{sec5} we give a general scheme for generating a full list of Lie Markov models with a given symmetry property,
In \S\ref{sec6} we explicitly give four state Lie Markov models with maximal symmetry.
Finally, \S\ref{conc} discusses implications and possibilities for future work.

\section{Lie algebras and closure of Markov models}\label{sec2}
For algebraic simplicity we work over the complex field $\mathbb{C}$, and refer to a matrix as ``Markov'' if it has unit column sums.
Later we will discuss how our discussion specializes to the stochastic case where the entries must be real and lie in the range $\left[0,1\right]$.
Rather than work directly with the general Markov model, we will also consider only Markov matrices that have non-zero determinant. 
Although this need not be the case for a general Markov matrix, it is not too stringent a condition as (i) the set of Markov matrices with zero determinant is of measure zero in the set of Markov matrices (this is because they are defined by the vanishing of a single polynomial function and hence lie in an ambient space of dimension one less than the set of generic Markov matrices), (ii) Markov matrices that arise from a continuous-time formulation have non-zero determinant (as we will see shortly).
In any case, in the conclusions we will argue that understanding Markov matrices with zero determinant becomes easier once we understand how the rest can be categorized.

Let the \emph{general Markov model} $\mathfrak{M}_{GMM}$ be the set of $n\times n$ matrices with column sum 1: 
\tb{\[\mathfrak{M}_{GMM}:=\left\{M\in \mathbb{M}_n(\mathbb{C}) : \bm{\theta}^T M =\bm{\theta}^T \right\},\]}
where $\bm{\theta}$  is the column $n$-vector with all its entries equal to $1$, ie. $\bm{\theta}^T=(1,1,\ldots,1)$.
Specializing further, consider the subset of matrices in $\mathfrak{M}_{GMM}$ with non-zero determinant:
\tb{
\[GL_1(n,\mathbb{C}):=\left\{M\in \mathbb{M}_n(\mathbb{C}) : \tb{\bm{\theta}^T} M =\tb{\bm{\theta}^T} ,\det(M)\!\neq \!0\right\}.\]}
In turn, this set of matrices includes a subset of matrices that arise by taking the exponential of a rate-matrix\tb{; that is, the exponential of a matrix in} 
\tb{
\begin{eqnarray} \label{eq:generalratematrixmodel}
\mathfrak{L}_{GMM}:=\left\{ Q\in M_n(\mathbb{C}) : \tb{\bm{\theta}^T} Q =\bm{0}^T \right\}.
\end{eqnarray}} 
We will refer to $e^{\mathfrak{L}_{GMM}}\tb{:=\left\{e^Q : Q\in \mathfrak{L}_{GMM}\right\}}$ as ``the general rate-matrix model'' and below we will discuss matrix exponentials in more detail (particularly their importance to Lie theory).

As the inverse of a Markov matrix (if it exists) is also a Markov matrix it is clear that $GL_1(n,\mathbb{C})$ is actually a subgroup of the general linear group $GL(n,\mathbb{C})$, and it follows that $GL_1(n,\mathbb{C})$ and $e^{\mathfrak{L}_{GMM}}$ are actually \emph{Lie groups} (see \cite{stillwell2008} for the relevant technical definitions).
In fact we have the isomorphism $GL_1(n,\mathbb{C})\cong A(n\!-\!1,\mathbb{C})$ where $A(n\!-\!1,\mathbb{C})$ is the (complex)  \emph{affine group} (see for example \citet{baker2003}).
This observation allows the general methods of Lie theory to be applied to understanding continuous-time Markov models; see \citet{johnson1985} and \citet{mourad2004} for general results and discussion, and \citet{sumner2008} for applications to phylogenetics.

Summarizing, we have the following set inclusions
\beqn
e^{\mathfrak{L}_{GMM}}&\subset GL_1(n,\mathbb{C}) \subset \mathfrak{M}_{GMM},\nonumber
\eqn
and Lie group hierarchy
\beqn
e^{\mathfrak{L}_{GMM}}& < GL_1(n,\mathbb{C}) < GL(n,\mathbb{C}).\nonumber
\eqn

We define a \emph{Markov model} $\mathfrak{M}$ by taking $\mathfrak{M}\subseteq \mathfrak{M}_{GMM}$ as some well defined subset of the general Markov model.
Similarly, a \emph{rate-matrix model} $e^{\mathfrak{L}}$ is defined by taking $\mathfrak{L}\subseteq \mathfrak{L}_{GMM}$ as some well defined subset of rate-matrices drawn from the general rate-matrix model and taking the set of exponentials thereof (as in (\ref{eq:generalratematrixmodel})).
It follows immediately from these definitions that all rate-matrix models are Markov models.
In what follows we are primarily interested in the case that $\mathfrak{M}=e^{\mathfrak{L}}$, and in this case we will abuse our terminology and refer to $\mathfrak{L}$ as a ``model''.
\begin{mydef}
A Markov model $\mathfrak{M}$ is said to be \emph{multiplicatively closed} if and only if for all $M_1,M_2\in \mathfrak{M}$ we also have $M_1M_2\in \mathfrak{M}$.
\end{mydef}
Of course, recalling that matrix multiplication is associative, this is exactly the statement that $\mathfrak{M}$ forms a $\emph{semigroup}$ under matrix multiplication.

Presently we explore the conditions for which we can expect a rate-matrix model to be closed.
Consider an extension of standard phylogenetic models where each edge $e$ of the tree has its own rate-matrix $Q_e$ chosen from some model $\mathfrak{L}$.
We can envisage this process arising under a model where at each branching event a new set of transition rates is chosen (a generalization of the ideas given in \cite{drummond2006}).
Now, if we remove a single taxon from our tree there is a standard marginalization procedure that will give us a new tree with one fewer taxa.
However, on this new tree there will now be an edge $e_{ab}$ which is the join of the edges $e_a\!=\!(u_a,v_a)$ and $e_b\!=\!(u_b,v_b)$ with $u_b=v_a$ from the original tree, that is $e_{ab}\!=\!(u_a,v_b)$ (see Figure~\ref{fig:margTree} for an illustration of this).
Under the marginalization, the transition matrix for the edge $e_{ab}$ will then be the product of the matrices from edge $e_a$ and $e_b$:
\[M_{e_{ab}}=M_{e_b}M_{e_a}=e^{Q_b\tau_b}e^{Q_a \tau_a},\]
where $Q_a,Q_b\in\mathfrak{L}$ and $\tau_a,\tau_b$ are the corresponding edge lengths.
Now, for $e^{\mathfrak{L}}$ to be closed we require  $M_{e_{ab}}=e^{Q_{ab}(\tau_{a}+\tau_{b})}$, with $Q_{ab}\in\mathfrak{L}$.

\begin{figure}[thb]
\centering
\psfrag{a}{$e^{Q_a\tau_a}$}\psfrag{b}{$e^{Q_b\tau_b}$}
\psfrag{x}{$u_a$}\psfrag{y}{$v_a=u_b$}\psfrag{z}{$v_b$}
\psfrag{c}{$e^{Q_{ab}(\tau_a+\tau_b)}$}
\includegraphics[scale=0.7]{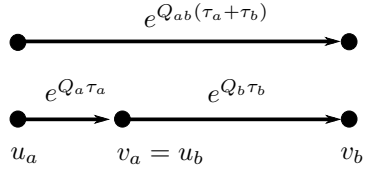}
\caption{$M_{e_{ab}}=e^{Q_{ab}(\tau_{a}+\tau_{b})}$, with $Q_{ab}\in\mathfrak{L}$.}
\label{fig:margTree}
\end{figure}

The question then naturally arises, is it obvious that this will be the case no matter how we define our model $\mathfrak{L}$?
To understand what could go wrong, consider two matrices $X,Y\in\mathbb{M}_n(\mathbb{C})$ and recall the classical Baker-Campbell-Hausdorff (BCH) formula \citep{campbell1897}:
\begin{eqnarray*}
 e^Xe^Y=\exp({X+Y+\fra{1}{2}\left[X,Y\right]+\fra{1}{12}\left[X,\left[X,Y\right]\right]+\ldots)},
\end{eqnarray*}
where $\left[X,Y\right]:=XY-YX$ is known as the \emph{commutator} (or \emph{Lie bracket}) of $X$ and $Y$, and the higher order terms are all given as further commutators of commutators of $X$ and $Y$ (see \cite{stillwell2008} for an elementary proof).
This formula generalizes the extremely well-known rule for the product of two exponentials, $e^xe^y=e^{x+y}$, from the case of commuting variables, $xy\!=\!yx$, to the more general case of non-commuting variables.
This is achieved by ``correcting'' for the non-commutivity of matrix multiplication with the addition of the further commutators.

Considering again the case at hand, if we replace $X$ and $Y$ with our rate-matrices $Q_a$ and $Q_b$, we see that a sufficient condition for a model $\mathfrak{L}$ to satisfy the required condition is for it to be a \emph{Lie algebra}:
\[\tb{\tau_a}Q_a+\tb{\tau_b}Q_b\text{ and }\left[Q_a,Q_b\right]\in\mathfrak{L},\]
for all $Q_a,Q_b\in \mathfrak{L}$ \tb{and $\tau_a,\tau_b$}.

\private{Jeremy: We had for all $\tau_a,\tau_b\in \mathbb{R}$ but I think it's best to sidestep the issue of what field we are talking about in these early comments.}

Suppose we have a model $\mathfrak{L}$ that forms a Lie algebra, and suppose that we choose a sequence of rate-matrices from it: $(Q_1,Q_2,\ldots,Q_m)$.
Taking parameters $\tau_1,\tau_2,\ldots,\tau_m$ with $t=\tau_1+\tau_2+\ldots +\tau_m$, we can  consider the inhomogeneous process given by 
\beqn
e^{Q_1\tau_1}e^{Q_2\tau_2}\ldots e^{Q_m\tau_m}.\nonumber
\eqn 
However, because $\mathfrak{L}$ forms a Lie algebra, we can conclude, via repeated application of the BCH formula, that there exists a rate-matrix $\widehat{Q}\in\mathfrak{L}$ that acts as a homogeneous average:
\beqn
M(t):=e^{\widehat{Q}t}=e^{Q_1\tau_1}e^{Q_2\tau_2}\ldots e^{Q_m\tau_m}.\nonumber
\eqn
Thus we see that if we restrict our attention to models that form Lie algebras, we are free to interpret fitting phylogenetic data as finding an average of the (possibly) true inhomogeneous process\footnote{Here the technical issue arises that the BCH formula does not guarantee that the series $X+Y+\lie{X}{Y}+\ldots $ will actually converge. Although this issue need not concern us in the present work because there will be \emph{some} radius of convergence for the series, we refer the concerned reader to \citet{blanes2004} for further analysis.}.
\private{Jeremy: I'm worried that this is not altogether clear. For example, can we actually be sure that$\widehat{Q}$ given by the BCH formula exists in the sense that the series converges? This issue seems to be discussed in \citet{blanes2004}}
\private{Jesus: I'll think about it and I'll try to get a copy of Blanes and Casas's paper. According to Stillwell, p.152, it seems there is no problem with the convergence in a neighbourhood of $\id$. On the other hand, in Wikipedia it's said that the map $exp: M_n(\CC) \rightarrow GL(n,\CC)$ is surjective. If this is true, there is some $\widehat{Q}$ such that $e^{Q_1\tau_1}e^{Q_2\tau_2}\ldots e^{Q_m\tau_m}=e^{\widehat{Q}}$. What else could $\widehat{Q}$ be but the matrix given by BCH-formula?}

\begin{figure}[bt]
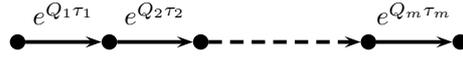

  \centering  
  \vspace{2em}
  $
  \psmatrix[colsep=1cm,rowsep=1cm,linewidth=.05cm]
	[mnode=dot,dotscale=1.2] & [mnode=dot,dotscale=1.2] &	[mnode=dot,dotscale=1.2] & & [mnode=dot,dotscale=1.2] & [mnode=dot,dotscale=1.2] \\
	\ncline{->}{1,1}{1,2}^{e^{Q_1\tau_1}}
	\ncline{->}{1,2}{1,3}^{e^{Q_2\tau_2}}
	\ncline{->}{1,5}{1,6}^{e^{Q_m\tau_m}}
	\psset{linestyle=dashed}\ncline{->}{1,3}{1,5}
	\endpsmatrix
  $
  \vspace{-3em}
  \caption{An inhomogeneous process.}
  \label{fig:inhomoprocess}
\end{figure}

For example, consider the general time reversible model (GTR) \citep{felsenstein2004}, with rate-matrices $Q$ satisfying $QD(\bm{\pi})=D(\bm{\pi})Q^T$ (i.e. $QD(\bm{\pi})$ is symmetric) where $\bm{\pi}$ is any (column) \emph{distribution vector}\footnote{We exclude $\pi_i\!=\!0$ in order to avoid some unimportant technicalities in what follows.}:
\[\bm{\pi}^T\!=\!(\pi_1,\pi_2,\ldots,\pi_n),\quad \pi_i>0,\quad \pi_1+\pi_2+\ldots+\pi_n\!=\!1,\] 
and $D(\bm{\pi})$ is the matrix with $\bm{\pi}$ on the diagonal and zeros elsewhere.
Taken over the complex field we can express the GTR model using the constraints
\begin{eqnarray}\label{eq:GTRcond}
\mathfrak{L}_{GTR}=\left\{Q \in \mathfrak{L}_{GMM} : \hspace{0em}\exists \hspace{0em} \bm{\pi} \text{ s.t. } QD(\bm{\pi})=D(\bm{\pi})Q^T \right\}.
\end{eqnarray}
Notice that, because $QD(\bm{\pi})=D(\bm{\pi})Q^T$ implies $Q^mD(\bm{\pi})=D(\bm{\pi})(Q^m)^T$ for all $m$, the GTR model expressed as transition matrices satisfies exactly the same conditions:
\[
e^{\mathfrak{L}_{GTR}}=\left\{M\in e^{\mathfrak{L}_{GMM}}: \hspace{0em}\exists \hspace{0em} \bm{\pi} \text{ s.t. } MD(\bm{\pi})=D(\bm{\pi})M^T \right\}.
\]

To show that the GTR model does not form a Lie algebra, we need the following result.
\begin{lem}\label{lem:posdiagcond}
If $X\in \mathbb{M}_n\left(\mathbb{C}\right)$ satisfies $XD(v)=-D(v)X$ for some vector $v=(v_1,v_2,\ldots, v_n)^T$, then, for all $i$ and $j$, either the matrix entry $X_{ij}=0$ \emph{or} $v_i=-v_j$.
\end{lem}
\begin{proof}
The condition on $X$ can be expressed in components as $X_{ij}v_j=-v_iX_{ij}$ for all $i$ and $j$.
The result follows directly by considering the individual cases.
\end{proof}
From now on, we denote the $n\times n$ identity matrix as $\id_n$ (or simply as $\id$ when $n$ is understood or unspecified).
Now, consider two rate matrices $Q_1$ and $Q_2$ that satisfy (\ref{eq:GTRcond}) for the uniform distribution $\bm{\pi}\!=\!\fra{1}{n}\bm{\theta}\!=\!\fra{1}{n}(1,1,\ldots,1)$.
As $D(\fra{1}{n}\bm{\theta})=\fra{1}{n}\id_n$, it is clear that the required condition is equivalent to demanding that $Q_1$ and $Q_2$ are symmetric.
If the rate-matrices of GTR model form a Lie algebra, the commutator $\lie{Q_1}{Q_2}$ must satisfy (\ref{eq:GTRcond}) for some (possibly different) distribution vector $\widehat{\bm{\pi}}$.
However,
\beqn
D(\widehat{\bm{\pi}})\lie{Q_1}{Q_2}^T=D(\widehat{\bm{\pi}})\left(Q_1Q_2-Q_2Q_1\right)^T=D(\widehat{\bm{\pi}})\left(Q_2Q_1-Q_1Q_2\right)=-D(\widehat{\bm{\pi}})\lie{Q_1}{Q_2},\nonumber
\eqn
where the second equality follows because $Q_1$ and $Q_2$ are symmetric by assumption.
Thus the GTR condition on the commutator becomes $\lie{Q_1}{Q_2}D(\widehat{\bm{\pi}})=-D(\widehat{\bm{\pi}})\lie{Q_1}{Q_2}$, and by Lemma~\ref{lem:posdiagcond} this is impossible \emph{unless} $\lie{Q_1}{Q_2}=0$.
\tr{As this is obviously not true in general (we will consider a particular example below), we see directly that}
\tb{To confirm that this is not true in general, consider the $n\times n$ rate-matrices
\beqn
Q_1=\left(\begin{array}{cccccc}
\ast & \alpha & \beta & 0 & \ldots & 0 \\
\alpha & \ast & 0 & 0 & \ldots & 0 \\
\beta & 0 & 0 & 0 & \ldots & 0 \\
0 & 0 & 0 & 0 &  & 0 \\
\vdots & \vdots   & \vdots  &  &  \ddots & \vdots \\
0 & 0 & 0 & 0 & \ldots & 0 
\end{array}
\right),
\qquad
Q_2=\left(\begin{array}{cccccc}
\ast & \alpha & \beta' & 0 & \ldots & 0 \\
\alpha & \ast & 0 & 0 & \ldots & 0 \\
\beta' & 0 & 0 & 0 & \ldots & 0 \\
0 & 0 & 0 & 0 &  & 0 \\
\vdots & \vdots & \vdots &  &  \ddots & \vdots  \\
0 & 0 & 0 & 0 & \ldots & 0 
\end{array}
\right)\nonumber,
\eqn
(where the missing entries are chosen so these matrices have zero column-sums).
Clearly $Q_1,Q_2\in \mathfrak{L}_{GTR}$ and $\lie{Q_1}{Q_2}=0$ if and only if $\beta\!=\!\beta'$.
Thus we conclude that}
\begin{res}\label{res:GTR}
The rate-matrices of the GTR model do not form a Lie algebra.
\end{res}

\tb{
\begin{res}\label{res:GTRnotClosed}
The GTR model is not multiplicatively closed.
\end{res}
\begin{proof}
First we recall that the Perron-Frobenius theorem (see for example \citet{PF}) claims that for each matrix $M$ with strictly positive matrix entries there is exactly one probability vector $\pi$ satisfying $M\pi=\pi$.
We consider two Markov matrices $M_1$ and $M_2$ the GTR condition for the uniform distibution, ie. they are symmetric: $M_1^T\!=\!M_1$ and $M_2^T\!=\!M_2$ (see above).
If GTR is multiplicatively closed there must exist a distribution vector $\widehat{\pi}$ such that the product $M_1M_2$ satisfies $M_1M_2D(\widehat{\pi})=D(\widehat{\pi})(M_1M_2)^T$.
However, $(M_1M_2)^T=M_2^TM_1^T=M_2M_1$, thus the required condition is
\beqn\label{eq:cond}
M_1M_2D(\widehat{\pi})=D(\widehat{\pi})M_2M_1.
\eqn  
Now $M_1M_2$ and $M_2M_1$ are both Markov matrices, so they satisfy
\beqn
\theta^TM_2M_1=\theta^T,\quad \theta^TM_1M_2=\theta^T.\nonumber
\eqn
Taking the transpose we have
\beqn\label{eq:PFuniform}
M_1M_2\theta=\theta,\quad M_2M_1 \theta= \theta.
\eqn
On the other hand, consider
\beqn\label{eq:PFnonuniform}
M_1M_2\widehat{\pi}=M_1M_2(D(\widehat{\pi})\theta)) 
&=D(\widehat{\pi})M_2M_1\theta\\
&=D(\widehat{\pi})\theta=\widehat{\pi}.
\eqn
Now we assume that $M_1$ and $M_2$ are chosen such that $M_1M_2$ satisfies the condition of the Perron-Frobenius theorem, and we see from (\ref{eq:PFuniform}) and (\ref{eq:PFnonuniform}) that we must have $\widehat{\pi}=\fra{1}{n}\theta$.
However, this implies that $D(\widehat{\pi})=\fra{1}{n}\id$ which from (\ref{eq:cond}) implies $M_1M_2 = M_2M_1$.
The proof is completed by finding particular examples of $M_1$ and $M_2$ that satisfy the required conditions and $M_1M_2 \neq M_2M_1$.
To this end consider the $n\times n$ Markov matrices:
\beqn
M_1=\left(\begin{array}{cccccc}
\ast & a & b & c & \ldots & c \\
a & \ast & c & c & \ldots & c \\
b & c & \ast & c & \ldots & c \\
c & c & c & \ast &  & c \\
\vdots & \vdots   & \vdots  &  &  \ddots & \vdots \\
c & c & c & c & \ldots & \ast 
\end{array}
\right),
\qquad
M_2=\left(\begin{array}{cccccc}
\ast & a & b' & c & \ldots & c \\
a & \ast & c & c & \ldots & c \\
b' & c & \ast & c & \ldots & c \\
c & c & c & \ast &  & c \\
\vdots & \vdots & \vdots &  &  \ddots & \vdots  \\
c & c & c & c & \ldots & \ast 
\end{array}
\right)\nonumber,
\eqn
where $a,b,c>0$, $a+b+(n-3)c<1$ (with similar for $b'$), and the missing entries are chosen so these matrices have unit column-sums.
These matrices are symmetric and their product satisfies the conditions of the Perron-Frobenius theorem.
Finally it is easy to confirm that $M_1M_2\neq M_2M_1$, as required.
\end{proof}
}

These results show that the GTR model is not multiplicatively closed and hence one cannot interpret a fit of the GTR model as a homogeneous average of a true inhomogeneous process.
This observation poses serious questions of interpretation for any phylogenetic inferences achieved using the GTR model.
It would be very interesting to perform an exploratory study of how much (or how little) the non-closure of GTR affects phylogenetic inference in practice, however this is outside the scope of our present discussion.

\begin{rem}
The reader may object that our definition (\ref{eq:GTRcond}) of the GTR model is too general and may prefer to define a different copy of ``the general time-reversible model'' for each fixed distribution vector $\bm{\pi}$:
\[\mathfrak{L}_{GTR_{\bm{\pi}}}:=\left\{Q \in \mathfrak{L}_{GMM} :  QD(\bm{\pi})=D(\bm{\pi})Q^T \right\}.\]
After all, every rate-matrix in $\mathfrak{L}_{GTR_{\bm{\pi}}}$ has stationary distribution $\bm{\pi}$: 
\[QD(\bm{\pi})=D(\bm{\pi})Q^T\Rightarrow Q\bm{\pi}=\mathbf{0},\]
and thus $\mathfrak{L}_{GTR_{\bm{\pi}}}$ seems to be a reasonably well motivated model.  
However, it is shown in \citet{jarvis2010} that $\mathfrak{L}_{GTR_{\bm{\pi}}}$ does not form a Lie algebra for any $\bm{\pi}$, \tr{explicit three-state example given above shows that $\mathfrak{L}_{GTR_{\bm{\pi}}}$}\tb{and the proof of Result~\ref{res:GTRnotClosed} shows that ${GTR_{\bm{\pi}}}$} with $\bm{\pi}$ the uniform distribution is definitely not closed.
\tb{Thus, it seems unlikely that $\mathfrak{L}_{GTR_{\bm{\pi}}}$ is closed for any choice of $\bm{\pi}$ (see below for a proof).}
In any case, in a practical context GTR is usually implemented by considering $\bm{\pi}\!=\!(\pi_1,\pi_2,\ldots,\pi_n)^T$ as providing ``free parameters'' that are inferred using the data at hand, and we therefore argue that it is the more general form of GTR model (\ref{eq:GTRcond}) that is most relevant.
\end{rem}

\tb{
\begin{res}
The $GTR_{\bm{\pi}}$ model is not closed for any distribution vector $\pi$.
\end{res}
\begin{proof}
If $GTR_{\bm{\pi}}$ is closed we must have $M_1M_2D(\pi)=D(\pi)(M_1M_2)^T=D(\pi)M_2^TM_1^T$ for all $M_1,M_2\in GTR_{\bm{\pi}}$. 
However $M_iD(\pi)=D(\pi)M_i^T$ then implies that $M_1M_2D(\pi)=M_2M_1D(\pi)$.
Since $\pi$ is a distribution vector (see the definition above), we conclude that $D(\pi)^{-1}$ exists, so we require $M_1M_2=M_2M_1$, ie. $GTR_{\bm{\pi}}$ is abelian.
To confirm that this cannot be the case, consider two rate matrices $Q_1,Q_2\in \mathfrak{L}_{GTR_{\bm{\pi}}}$ such that $\lie{Q_1}{Q_2}\neq 0$ and the two paths $A(t)=e^{tQ_1}$ and $B(t)=e^{tQ_2}$.
Now $C(s,t)=A(t)B(s)A(t)^{-1}\in GTR_{\bm{\pi}}$ for all $s,t$, but if $GTR_{\bm{\pi}}$ is abelian we have $C(s,t)=B(s)$, so $C(s,t)$ is independent of $t$.
On the other hand
\beqn
\frac{d}{d s}\left. \left( \left. \frac{d \, C(s,t)}{d t} \right |_{t=0}\right)\right|_{s=0}=A'(0)B'(0)-B'(0)A'(0)=\lie{Q_1}{Q_2}\neq 0,\nonumber
\eqn
which is a contradiction.
Constructing examples of $Q_1$ and $Q_2$ satisfying the required conditions completes the proof.
Consider the $n\times n$ rate matrices
\beqn
Q_1=\left(\begin{array}{cccccc}
\ast & \pi_1\alpha & \pi_1\beta & 0 & \ldots & 0 \\
\pi_2\alpha & \ast & 0 & 0 & \ldots & 0 \\
\pi_3\beta & 0 & 0 & 0 & \ldots & 0 \\
0 & 0 & 0 & 0 &  & 0 \\
\vdots & \vdots   & \vdots  &  &  \ddots & \vdots \\
0 & 0 & 0 & 0 & \ldots & 0 
\end{array}
\right),
\qquad
Q_2=\left(\begin{array}{cccccc}
\ast & \pi_1\alpha & \pi_1\beta' & 0 & \ldots & 0 \\
\pi_2\alpha & \ast & 0 & 0 & \ldots & 0 \\
\pi_3\beta' & 0 & 0 & 0 & \ldots & 0 \\
0 & 0 & 0 & 0 &  & 0 \\
\vdots & \vdots & \vdots &  &  \ddots & \vdots  \\
0 & 0 & 0 & 0 & \ldots & 0 
\end{array}
\right)\nonumber.
\eqn
It is easy to show that $Q_iD(\pi)=D(\pi)Q_i^T$ while $\lie{Q_1}{Q_2}=0$ if and only if $\beta'\!=\!\beta$.
\end{proof}
}

Presently we will recall some key definitions and results from elementary Lie theory, and go on to describe the Lie algebra associated with the general Markov model.
For the reader who is unfamiliar with the theory, we can recommend the elementary texts \citet{stillwell2008} and \citet{tapp2005}.

\begin{mydef}
A \emph{matrix group} $\mathcal{G}$ is a closed subgroup of $GL(n,\mathbb{C})$.
\end{mydef}

In this definition ``closed'' means in the topological sense. That is, if the limit of a sequence of matrices in $\mathcal{G}$ is non-singular, then the limit is also in $\mathcal{G}$.

\begin{mydef}\label{def:tangentspace}
The \emph{Lie algebra} of a matrix group $\mathcal{G}$ is the tangent space of $\mathcal{G}$ (regarded as a manifold) at the identity: $T_\id(\mathcal{G})$.
That is $X\in T_\id(\mathcal{G})$ is equivalent to the existence of a smooth path $A:\lie{0}{1}\rightarrow \mathcal{G}$ such that $A(0)=\id$ and $A'(0):=\left.\frac{dA\tb{(t)}}{dt}\right|_{t=0}=X$.
\end{mydef}
\private{OLD Jesus: I'm getting confused with the field. We work over the complex field and the Lie algebra defined above inherits complex structure. The dimension of $T_\id(\mathcal{G})$ as a $\mathbb{C}$-vector space equals the dimension of $\mathcal{G}$ (for example, $n^2$ if $\mathcal{G}=GL(n,\mathbb{C})$). Of course, we can consider $T_\id(\mathcal{G})$ as a real vector space, but then, the dimension is twice the dimension of $\mathcal{G}$. This is related to my comment after Lemma 2.4 below. }
\private{Jeremy: I have introduced complexification below. I hope it's ok now. It's a bit confusing but I'm not sure what we can do about this.}

With this definition it is not too hard to show that $T_{\id}(\mathcal{G})$ is a vector space over $\mathbb{R}$, because $X+rY$ is the tangent of the product of two smooth paths $A(t)B(rt)$, \tb{with $A(0)=\id$, $B(0)=\id$, $A'(0)\!=\!X$, $B'(0)\!=\!Y$ and $r\in \mathbb{R}$}.
Also, $T_{\id}(\mathcal{G})$ is closed under Lie brackets, because $C'(0)\!=\!\lie{X}{Y}$ is the tangent of the path $C(t)\!:=\!A(t)B'(0)A(t)^{-1}\in T_\id(\mathcal{G})$ and $C'(0)\in T_\id(\mathcal{G})$ because the tangent space includes all its limit points.

A standard (and powerful) tool in Lie theory is the exponential map:
\beqn
\exp: X\in\mathbb{M}_n(\mathbb{C})\mapsto \exp(X)=\id+X+\frac{X^2}{2!}+\frac{X^3}{3!}+\ldots ,\nonumber 
\eqn
which has infinite radius of convergence, so $e^X$ is defined for all $X$.
Recall that $\det(e^X)\!=\!e^{tr(X)}$, so we can conclude that $e^X\in GL(n,\mathbb{C})$ for all $X\in\mathbb{M}_n(\mathbb{C})$.
By defining the path $A(t)\!:=\!e^{Xt}$, we see immediately that \tb{$A(0)=\id$ and} $A'(0)=X$ so that $X\in T_{\id}(GL(n,\mathbb{C}))$.
Given that $X$ was chosen arbitrarily, it is clear that $T_\id(GL(n,\mathbb{C}))=\mathbb{M}_n(\mathbb{C})$.

We must however pause to reflect that although $T_\id(GL(n,\mathbb{C}))\!=\!\mathbb{M}_n(\mathbb{C})$, it follows from Definition~\ref{def:tangentspace} that $T_\id(GL(n,\mathbb{C}))$ is a vector space only over $\mathbb{R}$.
To clarify this point, observe that the elementary matrices $\{E_{ij}\}_{i, j}$, with matrix elements $\left(E_{ij}\right)_{kl}=\delta_{ik}\delta_{jl}$, form a basis for $\mathbb{M}_n(\mathbb{C})$ considered as a complex vector space, ie. $\mathbb{M}_n(\mathbb{C})=\lrc{\{E_{ij}\}_{i,j}}$. 
However, by definition $T_\id(GL(n,\mathbb{C}))$ is a real vector space and we see that a basis is given by the set $\{E_{ij},\bm{i}E_{kl}\}_{i, j,k,l}$, where $\bm{i}=\sqrt{-1}$, ie. $T_\id(GL(n,\mathbb{C}))=\lrr{\{E_{ij},\bm{i} E_{kl}\}_{i,j,k,l}}$.
These observations prompt the following definition:
\begin{mydef}
The \emph{complexification} of a (real) Lie algebra $T_\id(\mathcal{G})$ is the complex vector space $T_\id(\mathcal{G})^\mathbb{C}$ spanned by all linear combinations $c_1 X+c_2 Y$ with $X,Y\in T_\id(\mathcal{G})$ and $c_1,c_2\in\mathbb{C}$.
When $T_\id(\mathcal{G})=T_\id(\mathcal{G})^\mathbb{C}$ as \emph{sets} and $T_\id(\mathcal{G})=T_\id(\mathcal{G})^\mathbb{C}=\lrc{X_1,X_2,\ldots,X_k}$, where $\{X_i\}_{1\leq i\leq k}$ is a set of linearly independent tangent vectors over $\mathbb{C}$, we say that $\{X_i\}_{1\leq i\leq k}$ forms a \emph{$\mathbb{C}$-basis} of $T_\id(\mathcal{G})$.
\end{mydef}
\begin{rem}
It should be noted that it is not always the case that $T_\id(\mathcal{G})=T_\id(\mathcal{G})^\mathbb{C}$; in general this is only true if $\mathcal{G}$ is a manifold over $\mathbb{C}$.
Since this is the case for the matrix groups that we will consider in this article (by the definitions given at the start of this section), we will henceforth implicitly assume that complexification has been performed and that the Lie algebras we discuss are vector spaces over $\mathbb{C}$. 
\end{rem}

Presently we derive the Lie algebra of the general Markov model, as was first given in \citet{johnson1985}.
The commutator of any two elementary matrices can be checked explicitly:
\beqn\label{eq:Eijcom}
\left[E_{ij},E_{kl}\right]=E_{il}\delta_{jl}-E_{kj}\delta_{il}.
\eqn

We then define the elementary rate-matrices as
\begin{eqnarray*}
L_{ij}:=E_{ij}-E_{jj}, 
\end{eqnarray*}
for every $i\!\neq\! j$ and note that $\bm{\theta}^{T}L_{ij}=\bm{0}$ for every couple $(i,j)$ with $i\neq j$, so that these matrices are indeed rate-matrices.
It is then clear that we can express any rate-matrix $Q$ as a linear sum:
\beqn
Q=\sum_{\tb{i\neq j}}\alpha_{\tb{ij}}L_{ij},\nonumber
\eqn
and we have:
\begin{lem}\label{lem:GMMtangent}
The matrices $\{L_{ij}\}_{i\neq j}$ form a $\mathbb{C}$-basis for the tangent space of $GL_{1}(n,\mathbb{C})$.
\end{lem}
\begin{proof}
General Lie theory states that the dimension of the tangent space is equal to the dimension of the Lie group as a manifold.
Considering $GL_1(n,\mathbb{C})$ as a subgroup of $GL(n,\mathbb{C})$ it is clear that the dimension of $GL_1(n,\mathbb{C})$ as a manifold is $n(n-1)$.
Now, for each $i\!\neq \!j$, $L_{ij}$ is a tangent vector of the smooth path $A^{(ij)}(t):=e^{L_{ij}t}\in GL_1(n,\mathbb{C})$.
Also, there are $n(n-1)$ of the $L_{ij}$ and they are obviously linearly independent.
Therefore the tangent space of $GL_1(n,\mathbb{C})$ \tb{at the identity} is $\lr{L_{ij}}_\mathbb{C}$.  
\end{proof}

\begin{res}
The rate-matrices $\mathfrak{L}_{GMM}=\lrc{\{L_{ij}\}_{i\neq j}}$ form a Lie algebra.
\end{res}
\begin{proof}
The result follows directly as a consequence of Lemma~\ref{lem:GMMtangent}.
\end{proof}
Indeed, the Lie algebra structure of $\mathfrak{L}_{GMM}$ follows from (\ref{eq:Eijcom}) and a little manipulation with the elementary rate-matrices:
\beqn\label{eq:GMMalg}
\left[L_{ij},L_{kl}\right]=\left(L_{il}-L_{jl}\right)\left(\delta_{jk}-\delta_{jl}\right)-\left(L_{kj}-L_{lj}\right)\left(\delta_{il}-\delta_{jl}\right).
\eqn
For convenience in subsequent calculations we record a few individual cases of these commutation relations where we take $i,j,k,l$ all to be distinct:
\begin{center}
\begin{tabular}{lll}
$\left[L_{ij},L_{kl}\right]=\left[L_{ij},L_{il}\right]=0$, & $\left[L_{ij},L_{ki}\right]=L_{ij}-L_{kj}$, & $\left[L_{ij},L_{jl}\right]=L_{il}-L_{jl}$, \\[.5em] $\left[L_{ij},L_{kj}\right]=L_{kj}-L_{ij}$, & $\left[L_{ij},L_{ji}\right]=L_{ij}-L_{ji}$.
\end{tabular}
\end{center}

Of course, this is a very convenient basis for $\mathfrak{L}_{GMM}$ because any ``stochastic'' rate-matrix $Q$ (with real and non-negative off-diagonal entries) can be written as $Q=\sum_{i\neq j}\alpha_{ij}L_{ij}$ where the stochastic condition is simply that the coefficients $\alpha_{ij}$ are real and non-negative.
This prompts the following definition:

\begin{mydef}\label{stochasticbasis}
A Lie algebra $\mathfrak{L}\subset \mathfrak{L}_{GMM}$ has a \emph{stochastic basis} if there exists a basis $B_\LL=\{L_1,L_2,\ldots ,L_d\}$ of $\mathfrak{L}$ such that each $L_k$ is a convex linear combination of the $L_{ij}$, i.e.  $L_{k}=\sum_{i\neq j}\alpha_{ij}L_{ij}$ where $\alpha_{ij}\geq 0$. 
In this case, we say that $e^\mathfrak{L}$ is a \emph{Lie Markov model}.
\end{mydef}
\tb{In other words, each $L_k$ lies in the stochastic cone spanned by the vectors $\{L_{ij}\}_{i\neq j}$. }

\begin{mydef}
\tb{The \emph{dimension} of a Lie Markov model $e^\mathfrak{L}$ is the vector-space dimension of $\mathfrak{L}$.}
\end{mydef}

\begin{rem}
\tb{In most cases, this definition of dimension corresponds exactly to the number of ``free parameters'' of a given model.
For example, the dimension of the general Markov model is $n(n-1)$.}
\end{rem}

\section{Binary Lie Markov models}\label{sec3}
In this section we describe the two-state (or binary) Lie Markov models in full detail.
We will see that there is actually a continuous infinity of one-dimensional Lie Markov models in this case, and it is clear that this property generalizes to more states.
This will motivate us to consider the symmetry of the two-state models to enable us to give a classification, and demanding similar symmetry in the models with more states will be the key to significant progress in those cases.

\private{Jesus: maybe because of my poor English, but ``undesirable'' sounds a bit over the top to me. In any case, I think we should say something about why we consider this property as undesirable.}
\private{Jeremy: I have removed the undesirable word, it wasn't needed.}

Consider the general rate-matrix model on two states.
As is discussed in \citep{jarvis2010,sumner2010}, any $2\times 2$ rate-matrix can be expressed as the linear sum 
\beqn
Q=\alpha L_{\alpha}+\beta L_{\beta}=\left(
\begin{array}{rr}
	-\alpha & \beta \\
	\alpha & -\beta  \\
\end{array}\right).\nonumber
\eqn 
where
\beqn
L_{\alpha}:=L_{12}=\left(
\begin{array}{rr}
	-1 & 0 \\
	1  & 0 \\
\end{array}\right),
\qquad
L_{\beta}:=L_{21}=\left(
\begin{array}{rr}
	0 & 1 \\
	0  & -1 \\
\end{array}\right).\nonumber
\eqn 
Thus the Lie algebra of $GL_1(2,\mathbb{C})$ can be expressed as $\mathfrak{L}_{GMM}=\lrc{L_\alpha,L_\beta}$, with the non-trivial commutator $\left[L_\alpha,L_\beta\right]=L_\alpha-L_\beta$.
It is clear that $\mathfrak{L}_{GMM}$ equals the space of all $2\times 2$ rate-matrices, thus $\mathfrak{L}_{GMM}$ is the only two-state, two-dimensional Markov Lie algebra.

In the one-dimensional case, we can choose any tangent vector $L=\alpha L_\alpha+\beta L_\beta\in\mathfrak{L}_{GMM}$, with fixed $\alpha$ and $\beta$ not both zero, and take $\mathfrak{L}=\lrc{L}$.
The Lie algebra is trivial since the only Lie bracket is identically zero.
Additionally, two tangent vectors $L=\alpha L_\alpha+\beta L_\beta$ and $L'=\alpha' L_\alpha+\beta' L_\beta$ give the same model if $\lrc{L}=\lrc{L'}$, which occurs if $(\alpha,\beta)=(\lambda\alpha',\lambda\beta')$ for some complex number  $\lambda\neq 0$. 
Thus by varying $(\alpha,\beta)$ up to this scaling equivalence we are led to a complex projective space (that is, $\mathbb{C}\mathbb{P}_1$) of one-dimensional two-state Lie Markov models.

Now this is not particularly satisfying as we would at least like to identify the ``binary-symmetric model'', $\alpha\!=\!\beta$, as a special point in the $\mathbb{C}\mathbb{P}_1$ continuum of models.
It is clear that the binary-symmetric model has a certain type of symmetry that the other one-dimensional models do not share, and it is the exploration of this symmetry that will be of great assistance to us when we study models with more than two states.

We have been implicitly considering the finite set $\{1,2\}$ as the states of the continuous-time Markov chain on two states.
Consider the group of permutations of these two states: $\SG_2\cong \{e,(12)\}$.
Notice that there is an \emph{action} of $\SG_2$ on the $2\times 2$ rate-matrices \tb{$Q\in \mathfrak{L}_{GMM}$}:
\[Q\mapsto K_{\sigma}QK_{\sigma}^{-1},\quad \forall\sigma\in\SG_2,\]
where $K_{\sigma}$ is the $2\times 2$ permutation matrix representing $\sigma$.
We demand that the one-dimensional Lie algebra $\mathfrak{L}=\lr{L}_\mathbb{C}$ is \emph{invariant} under this action of $\SG_2$, ie.
\beqn
\lrc{K_\sigma LK_\sigma^{-1}}=\lrc{L},\quad \forall\sigma\in\SG_2.\nonumber
\eqn
Of course, the only non-trivial constraint occurs for the permutation $\sigma\!=\!(12)$, and we are led directly to the binary-symmetric model by noticing that $\alpha\!=\!\pm\beta$ are the only solutions of the equation $K_{(12)} L K_{(12)}^{-1}=K_{(12)} \left(\alpha L_{12}+\beta L_{21}\right) K_{(12)}^{-1}=\alpha L_{21}+\beta L_{12}=\mu L$, with $\mu\in \mathbb{C}$.
The solution $\alpha\!=\!\beta$ is exactly the binary-symmetric model, and we can reject the solution $\alpha\!=\!-\beta$ as this would mean that $L\propto L_\alpha-L_\beta$, which would not provide a stochastic basis for $\lrc{L}$.
This is the key to understanding how the binary-symmetric model sits as a special point in the $\mathbb{C}\mathbb{P}_1$ of one-dimensional Lie Markov models.

Further, the two-dimensional Lie algebra $\mathfrak{L}_{GMM}=\lrc{L_\alpha,L_\beta}$ is also invariant under the action of $\SG_2$.
This is because $K_{(12)}\lrc{L_\alpha,L_\beta}K^{-1}_{(12)}=\lrc{L_\beta,L_\alpha}=\lrc{L_\alpha,L_\beta}$.
The fact that this model has $\SG_2$ symmetry equates to the fact that, as measured by the model itself, neither of the two states of the Markov chain is distinguished from the other.

\begin{res}\label{res:2state}
In the case of two state continuous-time Markov models, there are exactly two Lie Markov models with $\SG_2$ symmetry:
\begin{enumerate}
\item The binary-symmetric model generated by $\lrc{L_\alpha+L_\beta}$ with dimension one.
\item The general Markov model generated by $\lrc{L_\alpha,L_\beta}$ with dimension two.
\end{enumerate}
\end{res}

Any other choice of one-dimensional model demands a choice of $\alpha$ and $\beta$, which from a practical point of view, is somewhat equivalent to using the general rate-matrix model and using some kind of inference to choose $\alpha$ and $\beta$.
Thus we find that our general approach of exploring Lie Markov models with a given symmetry has thus far achieved a satisfactory classification of two-state models. 
In the next section we explore this concept of model symmetry in more detail. 

\section{\tb{Permutation symmetries of Markov models}}\label{sec4}
Roughly speaking, we say that a model has a \emph{symmetry} if under a nucleotide permutation ``something'' doesn't change.
The purpose of this section is to discuss what this something should be.
\tb{In what follows we label nucleotides $A,C,G,T$ with the integers $1,2,3,4$ respectively.}

\subsection{Equivariant models}
Consider the graphical representation of the well-known Kimura 3ST model (K3ST) given in Figure~\ref{fig:K3STgraph}.
What this graph implies is that any rate-matrix chosen from the K3ST model has three \emph{free parameters} that must be statistically inferred using the data at hand. 
By considering the graph, the most obvious symmetry that this model can have is the permutation of the nucleotides that leave the graph invariant. 
It is well known that these permutations form the group
\beqn
\mathbb{Z}_2\times \mathbb{Z}_2 \cong \{e,(12)(34),(13)(24),(14)(23)\}.\nonumber
\eqn

\begin{figure}
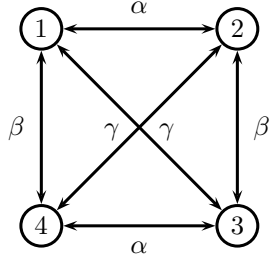

\centering
$
\psmatrix[colsep=2cm,rowsep=2cm,linewidth=.04cm,mnode=circle]
1&2\\
4&3
\ncline{<->}{1,1}{1,2}^\alpha
\ncline{<->}{2,1}{2,2}_\alpha
\ncline{<->}{1,1}{2,1}<\beta
\ncline{<->}{1,2}{2,2}>\beta
\ncline{<->}{1,1}{2,2}<\gamma
\ncline{<->}{1,2}{2,1}>\gamma
\endpsmatrix
$
\caption{Graphical representation of the K3ST model.}
\label{fig:K3STgraph}
\end{figure}

This observation motivates the following definition first given in \citet{draisma2008}.
\begin{mydef}
Given a group $G\leq \SG_n$, the $G$-\emph{equivariant model} denoted by $\mathbb{M}^G$ is defined as the algebraic evolutionary model obtained when taking transition matrices $M\in \mathbb{M}_n(\mathbb{C})$ such that  $K_{\sigma} M  K_{\sigma}^{-1}=M$ for every $\sigma \in G$ (here, $K_{\sigma}$ means the permutation matrix corresponding to $\sigma$). 
\end{mydef}

\begin{rem}
Notice that with this definition, the $G$-equivariant models are not Markov models.
\end{rem}
\begin{mydef}
Given a group $G\leq \SG_n$, we write $\mathfrak{M}^G=\mathbb{M}^G \cap GL_1(n,\CC)$ and we call it the $G$-\emph{equivariant Markov model}.  
We also write
$\mathfrak{L}^G=\{Q\in \mathfrak{L}_{GMM} : K_{\sigma} Q K_{\sigma}^{-1}=Q \}$ for the set of \emph{$G$-equivariant} rate-matrices. Notice that $\LL^G$ is a vector space. 
\end{mydef}

\private{Jeremy: I'm not sure why you have $\mathfrak{M}^G_1$. Compared to our def of $\mathfrak{M}$ I don't think the subscript is needed. I have removed it. }

\begin{lem} We have $T_{\id}(\mathfrak{M}^G)=\mathfrak{L}^G$.
\end{lem}

\begin{proof}
To prove this, we have to show that every matrix in $\mathfrak{L}^G$ is the derivative of some smooth path in $\mathfrak{M}^G$, and the converse. Now, if $X\in \LL^G$, consider $A(t)\!=\!e^{Xt}$. Clearly, $A(0)\!=\!\id$ and $A'(0)\!=\!X$. Moreover, $\det(A(t))=e^{t Tr(X)}\neq 0$ and if $\sigma\in G$, so 
\begin{eqnarray*}
K_{\sigma} A(t) K_{\sigma} ^{-1} & = &  K_{\sigma} \left (\sum_{i\geq 0}\frac{1}{i!}X^i t^i\right) K_{\sigma} ^{-1}= \sum_{i\geq 0}\frac{1}{i!}t^i (K_{\sigma}X^iK_{\sigma} ^{-1} ) = \\
& = &  \sum_{i\geq 0}\frac{1}{i!} X^i t^i=A(t). 
\end{eqnarray*}
This shows that the path $A(t)\subset \mathfrak{M}^G$ and proves one inclusion. 
For the converse, let $A(t)$ be any smooth path in $\mathfrak{M}^G$ with $A(0)\!=\!\id$. For any $\sigma\in G$, we have 
\begin{eqnarray*}
K_{\sigma} A(t) K_{\sigma} ^{-1}=A(t). 
\end{eqnarray*}
Taking the derivative, we see that 
\begin{eqnarray*}
K_{\sigma} A'(0) K_{\sigma} ^{-1}=A'(0),
\end{eqnarray*}
and we infer that the tangent vector $A'(0)$ is in $\LL^G$. From this, the other inclusion follows. 
\end{proof}

\begin{prop}Considered as continuous-time models \tb{$e^{\LL^G}$}, the ``equivariant models'' are Lie Markov models. 
\end{prop}

\begin{proof}
Let $\sigma \in G$ and let $X,Y\in \LL^G$. Then, we have  
\beqn
K_{\sigma}\left[X,Y\right]K_{\sigma}^{-1} =K_{\sigma}\left(XY-YX\right)K^{-1}_\sigma &=K_{\sigma}XK_{\sigma}^{-1}K_{\sigma}YK_{\sigma}^{-1}-K_{\sigma}YK_{\sigma}^{-1}K_{\sigma}XK_{\sigma}^{-1}\\
&=XY-YX\\
&=\left[X,Y\right],\nonumber
\eqn
 which shows that, considered as continuous-time models, the rate-matrices of equivariant models form Lie algebras.  
\end{proof}

\private{Jesus: I would like to state that $GL(n,\mathbb{C})\cap \mathfrak{M}_G$ is a Lie group. I think this is true, and then, we can present $\LL_G$ as the Lie algebra associated to it.}

\private{Jeremy: Yes we should do this. I think there is some general statement that subgroup of Lie group defined by some polynomial constraints is necessarily a Lie subgroup. I will try to chase it up, there are a few places we could use this result.}

\private{Jesus: Now, I think it would be interesting to know if every equivariant matrix $M\in GL(n,\mathbb{C})\cap \mathfrak{M}_G$ is the exponential of some $Q\in \LL_G$. $exp$ and $log$ are reciprocal inverse in a neighbourhood of $id$, but it seems that the $log$ series doesn't converge for any matrix $M$. However, according to wiki, the map $exp: M_n(\CC)  \rightarrow GL(n,\CC)$ is surjective (I haven't find a better reference, sorry) and if so, the transition matrix should be given by the logarithm series. In this case, I'm almost sure that the answer to my question above would by "yes". As I told Jeremy, I think this is interesting enough to be stated independently}

\private{Jesus: The above paragraph is related to Jeremy's question on the convergence of the BCH formula. }

To every Markov model, we associate an oriented graph with $n$ vertices $\{v_1,v_2,\ldots,v_n\}$ in the following way: vertices represent the possible states, one node for each state. Each ordered pair of vertices $(v_i,v_j)$ is connected by an arrow from $v_i$ to $v_j$ with a \tb{parameter} $\alpha_{ij}$ representing the rate that state $v_i$ changes into state $v_j$. 
In this way, every single rate-matrix in the model  provides values for these parameters, with different matrices providing different values.

Given such a graph, say $\Gamma$, and a permutation $\sigma\in\SG_n$, fix an ordering in the set of vertices. Then there is a natural action of $\sigma$ on $\Gamma$ defined by mapping $v_i$ to $v_{\sigma(i)}$. 
Notice that this induces an action of $\mathfrak{S}_n$ in the set of ordered pairs of vertices 
\begin{eqnarray*}
\sigma:(v_i,v_j)\mapsto  (v_{\sigma(i)},v_{\sigma(j)}),
\end{eqnarray*}
and we infer there is a group homomorphism
\begin{eqnarray*}
\rho:\mathfrak{S}_n \rightarrow \mathfrak{S}_{n(n-1)}.
\end{eqnarray*}
Under the action of such a permutation, we denote the new graph by $\sigma \Gamma$, and $\Gamma $ is said to be \emph{invariant} under $\sigma$ if $\sigma \Gamma= \Gamma$. This occurs if and only if
\begin{eqnarray}\label{rate}
\alpha_{ij}=\alpha_{\sigma(i)\sigma(j)}  \quad \forall v_i,v_j.
\end{eqnarray}

\begin{prop}\label{prop:GMMsymm}
If $K_{\sigma}$ is the $n\times n$ permutation matrix representing $\sigma\in \SG_n$ in the standard basis, the homomorphism $\rho:\SG_n\rightarrow \SG_{n(n-1)}$ described above in the case of $\mathfrak{L}_{GMM}$ is given by: 
\beqn\label{eq:GMMaction}
K_\sigma L_{ij}K_\sigma^{-1}=L_{\sigma(i)\sigma(j)}.
\eqn
\end{prop}

\begin{proof}
Consider the standard basis vectors $e_1,e_2,\ldots ,e_n$ of $\mathbb{C}^n$ with $K_\sigma e_i=e_{\sigma(i)}$ for all $\sigma\in\SG_4$.
It is easy to check that $L_{ij}e_k=\delta_{jk}(e_i-e_k)$ so $L_{\sigma(i)\sigma(j)}e_k=\delta_{\sigma(j)k}(e_{\sigma(i)}-e_k)$.
On the other hand $K_{\sigma}L_{ij}K^{-1}_{\sigma}e_k=K_\sigma L_{ij}e_{\sigma^{-1}(k)}=\delta_{j\sigma^{-1}(k)}\left(e_{\sigma(i)}-e_k\right)$.
Confirming that $\delta_{\sigma(j)k}=\delta_{j\sigma^{-1}(k)}$ for all $j$, $k$ and $\sigma\in \SG_4$ completes the proof.
\end{proof}

\begin{lem}
Let $\Gamma$ be the graph associated to a Markov model $\mathfrak{L}=\lrc{ L_1,\ldots,L_d}$, and let $\sigma \in \SG_4$ be a permutation. 
Then, $\Gamma$ is invariant under $\sigma$ if and only if $K_{\sigma}L_iK_{\sigma}^{-1}=L_i$, $\forall i\in \{1,2,\ldots,d\}$.
\end{lem}
\begin{proof}
Regarding the model $\mathfrak{L}$ as a vector subspace of $\mathfrak{L}_{GMM}$, the coordinates of the rate-matrices in $\mathfrak{L}$ in the basis $\{L_{ij}\}$ are just the mutation rates $\alpha_{ij}$: 
\begin{eqnarray}
Q=\sum_{i\neq j}\alpha_{ij}L_{ij}.\nonumber
\end{eqnarray}
Now, we have 
\begin{eqnarray*}
K_{\sigma}  Q K_{\sigma}^{-1}= \sum_{i\neq j}\alpha_{ij}K_{\sigma} L_{ij} K_{\sigma}^{-1}  = \sum_{i\neq j}\alpha_{ij} L_{\sigma(i)\sigma(j)}= \sum_{i\neq j}\alpha_{\sigma^{-1}(i)\sigma^{-1}(j)} L_{ij}.
\end{eqnarray*}
\private{Jeremy: Shouldn't it be $K_{\sigma} L_{ij} K_{\sigma}^{-1}=L_{\sigma(i)\sigma(j)}$??? I have changed it to this.}
By virtue of (\ref{rate}), it becomes clear that the invariance of the graph under $\sigma$ is equivalent to the invariance of every single rate-matrix of the model under $\sigma$. In particular, the invariance of the graph implies the invariance of $\{L_1,\ldots,L_d\}$ and this proves one implication. 
The other implication follows by using that $\{L_1,\ldots,L_d\}$ is a basis for $\mathfrak{L}$, so their invariance under $\sigma$ implies the invariance of every matrix in $\mathfrak{L}$. 
\end{proof}

\private{OLD Jeremy: I'm not completely convinced by the proof of this lemma. Is it possible to make it a little more clear? Jesus: Ok, I hope it's more clear now. Am I missing something? }
\private{Jeremy: OK, I think I get it now. But we have to sort out $\sigma$ vs $\sigma^{-1}$.}
\private{Jeremy: I have replaced ``evolutionary model'' above with ``Markov model''.}

\begin{res}\label{cor:equivariant}
Considered as continuous-time models \tb{$e^{\LL^G}$}, the equivariant condition corresponds exactly to the invariance of the graph associated to the model. 
ie. $\sigma\Gamma=\Gamma$ for all $\sigma\in G$.
\end{res}

\subsection{Statistical considerations}

\private{Jeremy: We need to let the reader know that we are back in the main discussion now. It's confusing that the following couple of paragraphs are jammed inbetween the equivariant and group-based subsection. I'm just not sure what to call this subsection!}

Corollary~\ref{cor:equivariant} nicely characterizes the equivariant models; a class of models that has allowed a number of already existing models to be analyzed simultaneously (for example in the work of \citet{draisma2008} and \citet{casanellas2010}).
However, from a statistical point of view, invariance of the graph is not particularly well motivated as a symmetry of a model.
Naturally, as the parameters of a model are ``free'', and hence need to be fitted using some statistical method and observed data, it follows that whatever inference method is used in practice, there will be no change to the outcome if the parameters themselves are permuted.

To make this statement more precise, suppose $F$ represents the maximum likelihood method of fitting a given (without loss of generality) two-dimensional continuous-time Markov model with rate-matrix $Q=\alpha_1 L_1+\alpha_2 L_2$ to a binary tree $\mathcal{T}$ with edge weights $\theta$ given some data $\mathfrak{D}$.
Thus $F$ is a function that returns (via optimization) maximum likelihood estimates of the free parameters in the model, i.e. $F(\mathfrak{D})\!=\!(\hat{\alpha_1},\hat{\alpha_2},\hat{\theta})$.
Now, if we permute the rate parameters in $Q$ by defining $Q'\!=\!\alpha_2 L_1+\alpha_1 L_2$, we claim that the corresponding modified function $F'$ will return precisely the same maximum likelihood estimates as $F$, i.e. $F'(\mathfrak{D})\!=\!(\hat{\alpha_1},\hat{\alpha_2},\hat{\theta})$. 
This is because the difference in the two is in parameter labels only and the optimization routine performed to find the maximum likelihood estimates will be unaffected by this.
This observation generalizes immediately to models with a greater number of free parameters and leads to the following characterization of symmetry:

\begin{mydef}\label{def:symm}
We say that a Lie Markov model $\mathfrak{L}$  has the \emph{symmetry} of the group $G\leq\SG_n$ if there is a basis $B_{\mathfrak{L}}=\{L_1,L_2,\ldots, L_d\}$ of $\LL$ such that 
\beqn
\sigma\cdot B_{\mathfrak{L}}:=\left\{K_{\sigma}L_{1}K^{-1}_{\sigma},K_{\sigma}L_{2}K^{-1}_{\sigma},\ldots ,K_{\sigma}L_{d}K^{-1}_{\sigma}\right\}=B_{\mathfrak{L}},\quad\forall \sigma\in G,\nonumber
\eqn 
and $G$ is the largest subgroup of $\SG_n$ with this property.
\end{mydef}

This definition means that $G$ acts by simply permuting the elements of a basis $B_{\mathfrak{L}}$ . 
Crucially, once such a basis is fixed, there is a group homomorphism $\rho:G\leq \mathfrak{S}_n \rightarrow \mathfrak{S}_d$, where $d$ is the dimension of the model.
That is, for all $\sigma\in G$ and $1\leq i \leq d$, we have
\beqn
K_{\sigma}L_{i}K^{-1}_{\sigma}=L_{\rho(\sigma)(i)},\nonumber
\eqn
where $\rho(\sigma)\in \SG_d$.
Also, the definition means that $\mathfrak{L}$ is invariant when considered as a \emph{vector space}: 
\beqn
\mathfrak{L}=\lrc{L_1,L_2,\ldots, L_d}\mapsto \sigma\cdot \mathfrak{L}:= \lrc{K_{\sigma}L_{1}K^{-1}_{\sigma},K_{\sigma}L_{2}K^{-1}_{\sigma},\ldots ,K_{\sigma}L_{d}K^{-1}_{\sigma}}=\mathfrak{L}.\nonumber
\eqn
However this is weaker than the condition given in the definition and therefore should not be seen as equivalent\footnote{Although the weaker vector space condition is consistent with the statistical motivations outlined above, we use the stronger condition given in Definition~\ref{def:symm} as it greatly simplifies the search for Lie Markov models.}.

For nucleotide models with ``maximal'' symmetry $\SG_4$, we see that, as measured by the model itself, there is no way of placing the nucleotides into any preferred groupings.
Indeed, any statistical inference method using such models will return the same answer no matter how the nucleotides are permuted (because such a permutation can be accounted for by a corresponding permutation of parameters).
It is somewhat surprising to learn (for the authors at least) that the largest symmetry of the Kimura 3ST model is $\SG_4$ itself (see \S\ref{sec6} below).

\begin{res}
The general $n$-state Markov Lie algebra $\LL_{GMM}$ has $\SG_n$ symmetry.
\end{res}
\begin{proof}
$\LL_{GMM}$ has basis $B_{GMM}=\{L_{ij}\}_{i\neq j}$.
Recalling Proposition~\ref{prop:GMMsymm}, we see that
\beqn
\sigma\cdot B_{GMM}=\{L_{\sigma(i)\sigma(j)}\}_{i\neq j}=\{L_{ij}\}_{i\neq j}=B_{GMM},\quad\forall\sigma\in\SG_n.\nonumber
\eqn
\end{proof}

\subsection{Group-based models}
Before proceeding to our general scheme, we end this section by showing that the group-based models are examples of Lie Markov models.

Given an abelian group $G$ of order $\left|G\right|=n$, a group-based model is defined by considering the $n$ group elements as the states of a continuous-time Markov chain with non-diagonal rate-matrix elements $\left[Q\right]_{ab}$ depending only on the difference $b-a$ (where the group operation is written as addition and its inverse is written as subtraction).
That is, we can write $\left[Q\right]_{ab}=\left[Q\right]_{b-a}$ for all $a,b\in G$.
As the possible values of the differences $b-a$ covers the whole of the group, by setting $\sigma=b-a$ we see that these models have $n$ free parameters $\alpha_{\sigma}:=\left[Q\right]_{ab}=\left[Q\right]_{b-a}$.
The interest in these group-based models is that they give rise to models with inherent symmetry (for example the Kimura 3ST model is exactly the group-based model produced by taking the abelian group $G=\mathbb{Z}_2\times \mathbb{Z}_2<\SG_4$).

Presently we will show that the group-based models form Lie Markov models by using standard results from the representation theory of finite groups.
We recommend \citet{sagan2001} as an elementary text for the reader unfamiliar with the basic theory.

\begin{mydef}
A \emph{representation} of a group $G$ is a map $\rho:G\rightarrow GL(V)\cong GL(m,\mathbb{C})$, where $V\cong \mathbb{C}^m$ and $\rho(g_1g_2)=\rho(g_1)\rho(g_2)$ for all $g_1,g_2\in G$.
We say that $\rho$ provides an \emph{action} of $G$ on the vector space $V$, and that $V$ forms a $G$-\emph{module} (or simply a \emph{module} when $G$ is clear from the context).
\end{mydef}

Given a rate-matrix $Q$ taken from a group-based model, consider the matrix derivatives:
\beqn
L_\sigma:=\left.\frac{\partial Q}{\partial\alpha_\sigma}\right|_{\underline{\alpha}=0}.\nonumber
\eqn
Notice that we can write $Q=\sum_{\sigma\in G}\alpha_\sigma L_\sigma$, and observe that $L_\sigma$ has matrix elements
\beqn\label{eq:Lelements}
\left[L_{\sigma}\right]_{ab}=\left\{
\begin{array}{l}
	1\text{, if }b-a=\sigma\\
	-1\text{, if }b-a=0\\
	0\text{, otherwise.}\\
\end{array}\right.
\eqn

\begin{mydef}\label{def:regrep}
The \emph{regular representation} of a group $G$ with elements $\{\sigma_1,\sigma_2,\ldots,\sigma_n\}$ is defined by taking the $n$-dimensional $G$-module 
\begin{eqnarray}
\lrc{G}=\lrc{\sigma_1,\sigma_2,\ldots,\sigma_n}=\{v=v_1\sigma_1+v_2\sigma_2+\ldots +v_n\sigma_n : c_i\in \mathbb{C}\},\nonumber
\end{eqnarray}
and group action defined as $v\mapsto \sigma\cdot v=v_1(\sigma\sigma_1)+v_2(\sigma\sigma_2)+\ldots+v_n(\sigma\sigma_n)$.
\end{mydef}

As, from Cayley's theorem, a group acts on itself by permutation, it is clear that the regular representation gives rise to permutation matrices. 
As we did in the case of permutation groups, we will denote the permutation matrix corresponding to $\sigma \in G$ by $K_{\sigma}$, so that the map $v \mapsto \sigma v$ is written as 
\begin{eqnarray}
v=\left(
\begin{array}{c}
  v_1  \\ v_2   \\ \ldots \\ v_n   
\end{array}
\right )
 \mapsto 
K_{\sigma} \left(
\begin{array}{c}
  v_1  \\ v_2   \\ \ldots \\ v_n   
\end{array}
\right )\nonumber
\end{eqnarray}

If we label the entries of $K_{\sigma}$ by the elements of the group $G$, it follows that
\beqn\label{eq:Kelements}
\left[K_{\sigma}\right]_{ab}=\left\{
\begin{array}{l}
	1\text{, if }b-a=\sigma\\
	0\text{, otherwise.}\\
\end{array}\right.
\eqn
Comparing (\ref{eq:Kelements}) and (\ref{eq:Lelements}) it is clear that:
\begin{lem}\label{lem:grouptangents}
The tangent vectors of a group-based model are given by $L_{\sigma}=-\id+K_{\sigma}$.  
\end{lem}

Observing that $L_{\sigma}=-\id+K_{\sigma}$ is a rate-matrix for \emph{any} $\sigma\in G$, Lemma~\ref{lem:grouptangents} allows us to extend the concept of group-based models to the non-abelian case:
\begin{mydef}
Given a (possibly non-abelian) group $G$, the corresponding \emph{group-based} model is given by $\mathfrak{L}_{G}:=\lrc{\{L_\sigma\}_{\sigma\in G}}$, where $L_\sigma=-\id+K_\sigma$ and $K_\sigma$ is the permutation matrix representing $\sigma$ in the regular representation (see Definition~\ref{def:regrep}).
\end{mydef}

\begin{prop}
The group-based model  $\mathfrak{L}_{G}=\lrc{\{L_\sigma\}_{\sigma\in G}}$ is a Lie Markov model.
\end{prop}
\begin{proof}
Consider the commutator of two tangent vectors arising from a group-based model:
\beqn
\lie{L_\sigma}{L_{\sigma'}}=\lie{-\id+K_{\sigma}}{-\id+K_{\sigma'}}=\lie{K_{\sigma}}{K_{\sigma'}}=K_{\sigma\sigma'}-K_{\sigma'\sigma}=L_{\sigma\sigma'}-L_{\sigma'\sigma}.\nonumber
\eqn
Thus the tangent vectors of a group-based model are closed under the operation of taking Lie brackets and hence form a Lie algebra.
\end{proof}

\section{General scheme \tb{for producing Lie Markov models}}\label{sec5}

Recall that in \S\ref{sec3} we observed that there is an $\mathbb{C}\mathbb{P}_1$ continuum of two-state, one-dimensional Lie Markov models.
This is an early indication that classifying arbitrary Lie Markov models is a difficult task. 
On the one hand, the problem seems to simply be to find all sub-algebras of $\mathfrak{L}_{GMM}$; however it must be kept in mind that we also need to ensure that these sub-algebras have a stochastic basis (as in Definition~\ref{stochasticbasis}).
We cannot use the classical Killing-Cartan-Dynkin classification of semi-simple Lie algebras (see \citet{erdmann2006}), as these results rely on isomorphism classes over $\mathbb{C}$ and this is completely incompatible with the concept of a stochastic basis. 
Thus producing a classification of Lie Markov models appears to be rather non-trivial and would rely on careful considerations of the geometry of the Lie bracket operation when restricted to a stochastic cone (cf. Definition~\ref{stochasticbasis}).

However, we also learnt in \S\ref{sec3} that the search for Lie Markov models can be significantly simplified by demanding that the models have symmetry (this successively reduced an infinite continuum of two-dimensional models to just two special cases).
In what follows we rely heavily on using symmetry to assist in the search for Lie Markov models.
Of course, it is expected that the larger the symmetry we demand, the easier the analysis will be. 

\subsection{Background on Group representation theory }
In what follows we recall and implement basic results from the representation theory of the symmetric group $\SG_n$.
Again, we recommend \citet{sagan2001} as an excellent introduction to the required material. 

Recall that a \emph{partition} of $n$ is a sequence of non-negative integers $\lambda\!=\!(\lambda_1,\lambda_2,\ldots ,\lambda_r)$, where $\lambda_{i}\geq \lambda_{i+1}$ for all $1\leq i\leq r$ and $\sum_{i=1}^r\lambda_i=n$.
We sometimes write $\lambda=\{\lambda_1^{n_1}\lambda_2^{n_2}\ldots \lambda_s^{n_s}\}$ to denote the partition that has $n_i$ copies of the integer $\lambda_i$, $1\leq i\leq s$.
For example, $\lambda=(5,5,4,2,2,1)\!=\!\{5^242^21\}$ is a partition of 19.

Recall also that a representation is said to be \emph{irreducible} if it does not contain any $G$- submodule. 

\begin{lem}
The irreducible representations of $\SG_n$ are in one-to-one correspondence with the partitions of $n$.
\end{lem}
We write $\rho_\lambda:\SG_n\rightarrow GL(V^\lambda)\cong GL(m,\mathbb{C})$ for the irreducible representation corresponding to the partition $\lambda$ and $V^\lambda\cong \mathbb{C}^m$ is the module carrying the representation.
In what follows, we will abuse notation and use the exponent notation to refer to both the partition itself \emph{and} to the $\SG_n$-module $V^\lambda$ that carries the representation $\rho_\lambda$.

\private{Jesus: In the following paragraph, I would mention ``Maschke's theorem''. Instead of (Sagan, 2001) I would cite a classical book, like Serre's. /  Another thing: I changed a little bit the order of explanation. Now, the first paragraph recalls the basics on rep. theory while the 2nd paragraph focuses on the $n=4$ case.}
Suppose we have a representation $\rho:\SG_n\rightarrow GL(V)$, for some complex vector space $V$.
As $\SG_n$ is finite, this representation is completely reducible into irreducible parts.
Hence we may write (Maschke's theorem):
\beqn
V\cong \oplus_{\lambda}c_\lambda V^\lambda,\nonumber
\eqn
where the sum is over all partitions $\lambda$ of $n$, and the $c_\lambda$ are non-negative integers specifying the number of copies of the irreducible module $V^\lambda$ that appear in the decomposition.
For example, recall that the \emph{defining} representation is given by the action of $\SG_n$ on the $n$-dimensional vector space $\mathbb{C}^n\cong \lrc{\{e_i\}_{1\leq i\leq n}}$ defined by $\sigma: e_i\mapsto e_{\sigma(i)}$.
It is a well known result that the defining representation decomposes as $\{n\}\oplus \{n-1,1\}$, where $\{n\}$ is the one-dimensional trivial representation and the irreducible representation $\{n-1,1\}$ therefore has dimension $(n-1)$.
Consider the \emph{projection operators}:
\[\Theta_{\lambda}:=\fra{1}{|\SG_n|}\sum_{\sigma \in\SG_n}\chi^{\lambda}(\sigma)\sigma,\]
where $\chi^\lambda$ is the character of the irreducible representation $\lambda$.
We recall that these operators project a given module onto its irreducible parts, i.e. $\Theta_{\lambda}V=c_\lambda V^\lambda$.
In this way we can use the $\Theta_{\lambda}$ to compute the $c_\lambda$.

\subsection{The general procedure}
Suppose we have a Markov Lie algebra $\mathfrak{L}$ and a permutation group $G\leq \SG_n$. 
We proceed by exploiting the action of $G$ on $\mathfrak{L}$ considered first as a \emph{discrete} structure, via a choice of basis $B_{\LL}=\{L_1,L_2,\ldots ,L_d\}$, and secondly as a \emph{linear} structure; that is, the vector space $\mathfrak{L}=\lrc{L_1,L_2,\ldots,L_d}$.

We demand that $\mathfrak{L}$ satisfies the conditions of Definition~\ref{def:symm} for the permutation group $G$.
Thus an action of $G$ is defined on the basis $B_{\LL}=\{L_1,L_2,\ldots ,L_d\}$ by $\sigma\in G:L_i \mapsto L_{\rho(\sigma)(i)}$, where $\rho$ is the homomorphism $\rho:G\rightarrow \SG_d$.
An \emph{orbit} of this action is a subset $\mathcal{B}=\{L_{a_1},L_{a_2},\ldots , L_{a_{|\mathcal{B}|}}\}\subset B_{\LL}$ such that $\mathcal{B}$ is invariant: \[\sigma\mathcal{B}:=\{L_{\rho(\sigma)(a_1)},L_{\rho(\sigma)(a_2)},\ldots , L_{\rho(\sigma)(a_{|\mathcal{B}|})}\}=\mathcal{B},\] 
for all $\sigma\in G$ and $\mathcal{B}$ contains no smaller subsets with this property (i.e. $\mathcal{B}$ is minimal).
Notice that this defines an equivalence relation that decomposes $B_\LL$ into disjoint orbits of $G$:
\begin{eqnarray*}
B_\mathfrak{L}=\mathcal{B}_1\cup \mathcal{B}_2 \cup \ldots \cup \mathcal{B}_r,
\end{eqnarray*}
where $\sigma\mathcal{B}_i=\mathcal{B}_i$ for all $i$ and $\sigma\in G$.

\tb{Now, it is a remarkable result of group actions that, up to bijective correspondence, a complete list of the orbits of a given group $G$ can be expressed using the \emph{orbit stabilizer theorem} (see \citet{bogopolski2008} for example), as follows.}
Consider a subgroup $H\leq G$ and partition $G$ into disjoint \tb{(right)} cosets $G/H=\{eH=H,\sigma_2 H,\ldots, \sigma_q H\}$ where each $\sigma_i\in G$ is chosen such that $\sigma_iH\!=\!\sigma_j H\Leftrightarrow i\!=\!j$ \tb{and $q\!=\!|G/H|\!=\!|G|/|H|$}.
Thus $G/H$ is a finite set with an action of $G$ defined by $\sigma:\sigma_i H\mapsto (\sigma\sigma_i) H$.
The orbit stabilizer theorem then says that there is a bijection of any orbit $\mathcal{B}$ with $G/G_x$, where $G_x=\{g\in G:gx=x\}$ is the \emph{stabilizer} of some element $x\in \mathcal{B}$. 
As $G_x\leq G$, and there are only finitely many subgroups of $G$, it is thus possible to give a complete list of the orbits of $G$ (up to isomorphism) by simply listing all $G/H$ with $H\leq G$.

Secondly, we can consider $\mathfrak{L}_{GMM}=\{Q=\sum_{i\neq j}\alpha_{ij}L_{ij}:\alpha_{ij}\in \mathbb{C}\}$ as a $G$-module by linearly extending the action (\ref{eq:GMMaction}):
\tb{
\beqn
Q=\sum_{i\neq j}\alpha_{ij}L_{ij}\mapsto \sigma\cdot Q=\sum_{i\neq j}\alpha_{ij}L_{\sigma(i)\sigma(j)},\nonumber
\eqn
for all $\sigma\in G$.}
By virtue of Maschke's theorem, we can also decompose $\mathfrak{L}_{GMM}$ into \emph{irreducible} $G$-modules $V^\lambda\subset \mathfrak{L}$ of the linear action of $G$.

\tb{On the other hand, for each $H\leq G$ we can extend the set $G/H$ to a complex vector space} 
\tb{\[\lrc{G/H}=\lrc{H,\sigma_2 H,\ldots, \sigma_q H}=\{v=c_1[e]+c_2[\sigma_2]+\ldots +c_q[\sigma_q]:c_i\in\mathbb{C}\},\] 
where $[\sigma]\equiv \sigma H$}, \tb{and consider this vector space as $G$-module via the mapping} 
\[\sigma:v=c_1[e]+c_2[\sigma_2]+\ldots +c_q[\sigma_q]\mapsto v'=c_1[\sigma]+c_2[\sigma\sigma_2]+\ldots +c_q[\sigma\sigma_q].\] 
\tb{We can then compare the irreducible $G$-modules that occur in the decomposition of $\mathfrak{L}_{GMM}$ to those that occur in the decomposition of $\lrc{G/H}$ for each $H\leq G$.
Finally, we can attempt to construct sub-algebras $\mathfrak{L}\subset \mathfrak{L}_{GMM}$ with a basis $B_\mathfrak{L}$ such that $B_\mathfrak{L}=\mathcal{B}_1\cup \mathcal{B}_2 \cup \ldots \cup \mathcal{B}_r$ is a plausible union of orbits $\mathcal{B}_i$ that are consistent with the linear decomposition of $\mathfrak{L}_{GMM}$} induced by the action of $G$.

\vspace*{3mm}
\noindent 
\textbf{Example.} Consider the action of $\SG_2\cong \{e,(12)\}$ on the two-state general Markov model considered as the finite set $B_{{GMM}}=\{L_{12},L_{21}\}$.
As $(12) L_{12}=L_{21}$ we see immediately that $B_{{GMM}}$ contains only one $\SG_2$-orbit (namely itself).
On the other hand, the subgroups of $\SG_2$ are the trivial group $H_1:=\{e\}$ and $H_2:=\SG_2$ itself.
We have the bijections: 
\begin{eqnarray}\label{bij1}
 \SG_2/H_1:=\{eH_1,\sigma H_1\}=\{\{e\},\{\sigma\}\}\mapsto \SG_2,
\end{eqnarray}
and 
\begin{eqnarray} \label{bij2}
\SG_2/H_2:=\{eH_1,\sigma H_1\}=\{\{e,\sigma\},\{\sigma,e\}\}=\{\{e,\sigma\}\}\mapsto\{e\}.
\end{eqnarray}

Thus, \tb{by comparing the cardinality of $B_{{GMM}}$ to these two orbits,} we conclude that $B_{{GMM}}\cong \SG_2/H_1 \cong \SG_2$, which is to say that the action of $\SG_2$ on $B_{{GMM}}$ is isomorphic to the action of $\SG_2$ on \emph{itself} (cf. Cayley's theorem).

Now, consider the action of $\SG_2$ on the two-state general Markov Lie algebra considered as a complex vector space: $\mathfrak{L}_{GMM}\!=\!\lrc{L_{12},L_{21}}$.
It is well known that there are exactly two irreducible $\SG_2$-modules \tb{$V^{id}$ and $V^{sgn}$} (both one-dimensional, $V^{id}\cong V^{sgn}\cong \mathbb{C}$), where the $\SG_2$ action is given by
\beqn
v\in V^{id}&\mapsto \sigma v=v,\nonumber\\ 
v\in V^{sgn}& \mapsto \sigma v=sgn(\sigma)v,
\eqn
for all $\sigma \in \SG_2$.
The orbit $\SG_2/H_1$ described in (\ref{bij1}) has \tb{size} two and linear decomposition $\lrc{\SG_2/H_1}\cong V^{id}\oplus V^{sgn}$, whereas the orbit in (\ref{bij2}) has \tb{size} one and has linear decomposition $\lrc{\SG_2/H_1}\cong V^{id}$.
If we define $L_{id}\!:=\!L_{12}+L_{21}$ and $L_{sgn}\!:=\!L_{12}-L_{21}$ we see immediately that $\lrc{L_{id}}\cong V^{id}$ and $\lrc{L_{sgn}}\cong V^{sgn}$.
Thus we conclude that $\mathfrak{L}_{GMM}\tb{=\langle L_{id}\rangle_{\CC}\oplus \langle L_{sgn}\rangle_{\CC}} \cong V^{id}\oplus V^{sgn}$, and we see that the only possible two-state Lie Markov models with $\SG_2$ symmetry are exactly the two cases given in \S\ref{sec3}. 
\vspace*{3mm}

\private{Jeremy: I replaced ``rank'' of orbit with ``size'' of orbit. I've also replaced $M_{id}$ with $V^{id}$. This is consistent with the notation in the next section -- $M$ is overused already!}

\tb{
These ideas can be generalized to models with a greater number of states.
The general procedure for generating a $n$-state Lie Markov model, $\mathfrak{L}$ , with $G\leq \SG_n$ symmetry is as follows. 
\begin{enumerate}
\item Decompose the Lie algebra of the general Markov model into irreducible modules of $G$ (Maschke's theorem): $\mathfrak{L}_{GMM}=\oplus_{\lambda}f_{\lambda} \mathfrak{L}^\lambda\cong \oplus_{\lambda}f_{\lambda}V^{\lambda}$, where $\lambda$ labels the irreducible $G$-module $V^\lambda\cong\mathfrak{L}^\lambda$ and the $f_\lambda$ are integers specifying how many times each irreducible module occurs in the decomposition.
\item Apply the orbit stabilizer theorem and construct the list of $G$-orbits, $G/H_i$, by working through the subgroups $H_i\leq G$. 
\item Extend each of the orbits linearly over $\mathbb{C}$ to the $G$-module $\lrc{G/H_i}$ and decompose each into irreducible $G$-modules: $\lrc{G/H_i}\cong \oplus_{\lambda}h^{(i)}_{\lambda}V^{\lambda}$, where again the $h^{(i)}_\lambda$ are integers.
\item Working up in dimension $d$, consider all unions of $G$-orbits $S:=(G/H_1)\cup (G/H_2)\cup \ldots \cup (G/H_q)$ such that $\left|S\right|=\sum_{1\leq i\leq q}\left|G/H_i\right|\!=\!d$ (where $|\cdot|$ stands for cardinality).
\item For each $S$, consider its linear decomposition into irreducible $G$-modules: $\lrc{S}\cong\oplus_\lambda a_\lambda V^\lambda$ where $a_\lambda:=h^{(1)}_\lambda+h^{(2)}_\lambda+\ldots +h^{(q)}_\lambda$, and, in order to exclude unions of $G$-orbits that do not occur in the linear decomposition of $\mathfrak{L}_{GMM}$ as a $G$-module, check that $a_\lambda\leq f_\lambda$, for each $\lambda$. 
\item For each case thus identified, consider the vector space $\mathfrak{L}\!:=\!\oplus_{\lambda}a_{\lambda}\mathfrak{L}^{\lambda}$ and use explicit computation to check whether $\mathfrak{L}$ forms a Lie algebra.
\item If $\mathfrak{L}$ forms a Lie algebra, attempt to show that it has a stochastic basis. 
\end{enumerate}
}

This procedure is guaranteed to produce all Lie Markov models with symmetry $G$ \tb{and is best understood by studying the examples given in the next section}.
\tb{We have successfully implemented it to determine the list of Lie Markov models} in the two-state case \tb{with $G=\SG_2$}, the three-state case with $G=\SG_3$, $\mathbb{Z}_3$ and the four state case with $G=\SG_4$, $\mathbb{Z}_2\wr \mathbb{Z}_2$, $\mathbb{Z}_2\times \mathbb{Z}_2$ and $\mathbb{Z}_4$.
In the \S\ref{sec6} we will give a complete presentation of the four state $G=\SG_4$ case, and defer presentation of the other cases to a future publication.
However, for $n$ ``large'' and $G$ ``small'', \tb{it is worth noting that} the final two steps in the procedure become quite difficult and computationally expensive.
Clearly, further theoretical ideas, such as those alluded to at the start of this section, are needed to describe Lie Markov models in their entirety.

\section{Lie Markov models with $\SG_4$ symmetry}\label{sec6}

As we are especially interested in nucleotide evolution, we fix $n\!=\!4$ and use the projection operators to decompose the Lie algebra of the general Markov model into irreducible representations of $\SG_4$.
The partitions of $n\!=\!4$ are $\{4\}$, $\{31\}$, $ \{2^2\}$, $\{21^2\}$ and $\{1^4\}$.
Each of these partitions labels an inequivalent irreducible representation of $\SG_4$ and the corresponding character table is given in Table~\ref{tab:chartabS4}.
Notice that the first row in the character table gives the dimension of each representation. Notice also that there are two one-dimensional representations, namely $\{4\}$, the \emph{trivial} representation where each permutation is mapped to the identity $1$, and $\{1^4\}$, the \emph{sign} representation where each permutation is mapped to $\pm 1$ based on the sign of the permutation.

\begin{table}[t]
\centering
\begin{tabular}{c|ccccc}
 & $\{4\}$  & $\{31\}$ & $\{2^2\}$ & $\{21^2\}$ & $\{1^4\}$ \\
\hline
$e$ & 1 & 3 & 2 & 3 & 1 \\
$[(12)]$ & 1  & 1 & 0 & -1\hspace{.35em} & -1\hspace{.35em}\\
$[(123)]$ & 1 & 0 & -1\hspace{.35em} & 0 & 1  \\
$[(12)(34)]$ & 1 & -1\hspace{.35em} & 2 & -1\hspace{.35em} & 1 \\
$[(1234)]$ & 1  & -1\hspace{.35em} & 0 & 1 & -1\hspace{.35em}\\
\end{tabular}
\caption{The character table of $\SG_4$. The rows are the conjugacy classes and the columns are the irreducible characters.}
\label{tab:chartabS4}
\end{table}

We explicitly confirm the decomposition of the defining representation for $n\!=\!4$ by constructing the projection operators
\beqn
\Theta_{\{4\}} &=\fra{1}{24}\sum_{\sigma \in\SG_4}\chi^{\{4\}}(\sigma)\sigma=\fra{1}{24}\left(e+(12)+(13)+(14)+(23)+\ldots +(1423)+(1432)\right),\nonumber\\
\Theta_{\{31\}} &=\fra{1}{24}\sum_{\sigma \in\SG_4}\chi^{\{31\}}(\sigma)\sigma\\
&=\fra{1}{24}\left(3e-(12)-(13)-(14)-(23)-(24) -(34) -(12)(34)-(13)(24)-(14)(23)\right.\\&\hspace{4em}\left.+(1234)+(1243)+(1324)+(1342) +(1423)+(1432)\right).
\eqn 
Notice that $\Theta_{\{4\}}\cdot e_i=\frac{1}{4}\left(e_1+e_2+e_3+e_4\right)$ for all $i$ and $\Theta_{\{31\}}\cdot e_1=\frac{1}{24}\left(6 e_1-2 e_2-2 e_3-2 e_4\right)$.
From this we conclude that the defining representation contains both the modules $\{4\}$ and $\{31\}$. 
The dimension count $4\!=\!1+3$ then shows that these are the only irreducible modules occurring in the defining representation (or, alternatively, one may check that $\Theta_{\{2^2\}}\cdot e_i=\Theta_{\{21^2\}}\cdot e_i=\Theta_{\{1^4\}}\cdot e_i=0$ for all $i$).
Hence:
\beqn
\lrc{\{e_i\}_{1\leq i\leq 4}}\cong \{4\}\oplus\{31\}.\nonumber
\eqn
We will use this decomposition of the defining representation repeatedly in what follows.

We would like to decompose $\mathfrak{L}_{GMM}=\langle \{L_{ij}\}_{1\leq i\neq j \leq 4} \rangle_{\mathbb{C}}$ into irreducible modules of $\SG_4$.
Recall that the action of $\SG_4$ is defined by $\sigma: L_{ij}\mapsto L_{\sigma(i)\sigma(j)}$. 
Compare this to the action of $\SG_4$ on the tensor product space $\mathbb{C}^4\otimes \mathbb{C}^4\cong \lrc{\{e_i\otimes e_j\}_{1\leq i, j \leq 4}}$ defined by $\sigma:e_i\otimes e_j\mapsto e_{\sigma(i)}\otimes e_{\sigma(j)}$. 
As any \tb{tangent vector} $L\in \mathfrak{L}_{GMM}$ can be expressed as $L=\sum_{1\leq i\neq j\leq 4}\alpha_{ij}L_{ij}$, we see immediately that the action of $\SG_4$ on $\mathfrak{L}_{GMM}$ is isomorphic to the action on the subspace of tensors $\left\{\psi\in \mathbb{C}^4 \otimes \mathbb{C}^4:\psi_{ii}=0,1\leq i\leq 4\right\}$.
If we disregard the constraints $\psi_{ii}=0$ for a moment, what we have is the Kronecker product of two copies of the defining representation: ie. $\mathbb{C}^4\cong\{4\}\oplus\{31\}$ and $\mathbb{C}^4\otimes \mathbb{C}^4\cong \left(\{4\}\oplus\{31\}\right)\otimes \left(\{4\}\oplus\{31\}\right)$.
Referring to Table~\ref{tab:chartabS4} and appealing to orthogonality of irreducible characters, we find that
\beqn\label{eq:decompGMM}
\mathbb{C}^4\otimes \mathbb{C}^4\cong\left(\{4\}\oplus\{31\}\right)\otimes \left(\{4\}\oplus\{31\}\right)=2\{4\}\oplus3\{31\}\oplus\{2^2\}\oplus\{21^2\}.\nonumber
\eqn 
Now the subspace spanned by the $e_i\otimes e_i$ is itself isomorphic to the defining representation: $\lrc{\{e_i\otimes e_i\}_{1\leq i\leq 4}}\cong \{4\}\oplus\{31\}$, and this subspace must appear in the decomposition of $\mathbb{C}^{4} \otimes \mathbb{C}^{4}$. 
Setting $\psi_{ii}=0$ is the same as removing this subspace from the decomposition.
Thus we have:
\begin{res}
The decomposition of the four state general rate-matrix model $\mathfrak{L}_{GMM}$ into irreducible representations of $\SG_4$ is given by 
\beqn\label{eq:schurdecompS4}
\mathfrak{L}_{GMM}\cong \{4\}\oplus2\{31\}\oplus\{2^2\}\oplus\{21^2\},
\eqn
where the decomposition of the dimension is given by $12\!=\!1+2\times 3+2+3$.
\end{res}

\subsection{A convenient basis}
We would now like to present an explicit basis for each module present in the decomposition (\ref{eq:schurdecompS4}).
Consider the vector
\beqn
L_{id}:=\sum_{{1\leq i\neq j\leq 4}}L_{ij}=L_{12}+L_{13}+\ldots +L_{43}=
\left(
\begin{array}{cccc}
	-3 & 1 & 1 &1  \\
	1  & -3 & 1 & 1 \\
	1 & 1 & -3 & 1 \\
	1 & 1 & 1 & -3
\end{array}\right).\nonumber
\eqn
Clearly this vector is invariant
\beqn
\sigma\cdot L_{id}=L_{id},\qquad \forall \sigma \in\SG_4,\nonumber
\eqn
and it hence spans the trivial representation: $\lrc{L_{id}}\cong \{4\}$.
Thus $\lrc{L_{id}}$ accounts for the first module appearing in the decomposition $(\ref{eq:schurdecompS4})$.

Define the \emph{row sum} vectors
\beqn
R_{i}:=\sum_{{j: 1\leq i\neq j\leq 4}}L_{ij},\nonumber
\eqn
and the corresponding \emph{column sum} vectors 
\beqn
C_{i}:=\sum_{{j: 1\leq i\neq j\leq 4}}L_{ji}.\nonumber
\eqn
For example, we have
\beqn
R_2=\left(
\begin{array}{cccc}
	-1 & 0 & 0 &0  \\
	1  & 0 & 1 & 1 \\
	0 & 0 & -1 & 0 \\
	0 & 0 & 0 & -1
\end{array}\right),
\qquad
C_4=\left(
\begin{array}{cccc}
	0 & 0 & 0 &1  \\
	0  & 0 & 0 & 1 \\
	0 & 0 & 0 & 1 \\
	0 & 0 & 0 & -3
\end{array}\right).\nonumber
\eqn
Consider the action of $\SG_4$ on each of these vectors:
\beqn
\sigma:R_{i}&\mapsto\sigma R_{i}=\sum_{{j:  1\leq i\neq j\leq 4}}L_{\sigma(i)\sigma(j)}=R_{\sigma(i)},\nonumber\\
\sigma:C_{i}&\mapsto\sigma C_i= \sum_{{j :  1\leq i\neq j\leq 4}}L_{\sigma(j)\sigma(i)}=C_{\sigma(i)}.
\eqn
Clearly these actions are isomorphic to the defining representation: $\sigma:e_i\mapsto e_{\sigma(i)}$.
Therefore, the (invariant) subspace generated by the row sum vectors, as well as that generated by the column sum vectors, is  isomorphic as a $\mathfrak{S}_4$-module to the defining representation:
\beqn\label{eq:Rdecomp}
\langle R_1,R_2,R_3,R_4\rangle_{\mathbb{C}}\cong\langle C_1,C_2,C_3,C_4 \rangle_{\mathbb{C}}\cong \{4\}\oplus \{31\}.
\eqn
Notice that these vectors are linearly independent except for the single linear relation
\beqn
R_1+R_2+R_3+R_4=C_1+C_2+C_3+C_4=L_{id}.\nonumber
\eqn
Keeping this linear dependence in mind, we may write
\beqn
\langle L_{id}\rangle_{\mathbb{C}}\oplus\langle R_1,R_2,R_3,R_4\rangle_{\mathbb{C}}\oplus\langle C_1,C_2,C_3,C_4 \rangle_{\mathbb{C}}\cong \{4\}\oplus 2\{31\},\nonumber
\eqn
and we see that we have now accounted for the first three modules occurring in the decomposition $(\ref{eq:schurdecompS4})$.

Referring to the graphical representation of the Kimura 3ST model given in Figure~\ref{fig:K3STgraph}, we see that the Kimura 3ST model has a basis given by:
\beqn
B_{K3ST}=\{ L_\alpha,L_\beta,L_\gamma\},\nonumber
\eqn
with 
\beqn
L_\alpha &=L_{12}+L_{21}+L_{34}+L_{43},\nonumber\\
L_\beta &=L_{13}+L_{31}+L_{24}+L_{42},\\
L_\gamma &=L_{14}+L_{41}+L_{23}+L_{32}.
\eqn
We claimed in \S\ref{sec4} that the symmetry of the Kimura $3ST$ model is all of $\SG_4$.

\private{Jeremy: I tried $B_{\mathfrak{L}_{K3ST}}$ and think it looks over the top. Do you think it is ok to use $B_{\mathfrak{L}}$ for an arbitrary model $\mathfrak{L}$ and $B_{GMM}$ and $B_{K3ST}$ for the particular models?}
\private{Jesus: here and somewhere in the previous sections, I suggest to replace "set isomorphism" for "bijection". It sounds much more natural to me and I think they are exactly the same thing in our framework.}
\private{Jeremy: Yes, good idea. It's actually pretty tricky! What we really need is different symbols $=$, $\cong$ etc depending on what structures we are claiming are equal/isomorphic.}

We now prove this:
\begin{res}\label{eq:K3STsymm}
The symmetry of the Kimura 3ST model is $\SG_4$.
\end{res}
\begin{proof}
We consider the following set of bipartitions of the set $\{1,2,3,4\}$: 
\beqn
S=\{12|34,13|24,14|23\},\nonumber
\eqn
where $ij|kl:=\{\{i,j\},\{k,l\}\}$.
Now take the bijection between these bipartitions and the three tangent vectors of the K3ST model given by:
\beqn\label{eq:isobik3st}
ij|kl\mapsto L_{ij}+L_{ji}+L_{kl}+L_{lk}.
\eqn
\private{Jeremy: I notice that you replaced $ij|kl$ with $ij\mid kl$ using the ``mid'' command. I think the later leaves too large a gap and prefer the look of the former. I have changed them back!}
This map is well-defined and, in particular,
\beqn
12|34\mapsto L_\alpha,\qquad 13|24\mapsto L_\beta,\qquad 14|23\mapsto L_\gamma,\nonumber
\eqn
and we note relations such as $12|34=43|12\mapsto L_{43}+L_{34} +L_{12}+L_{21}=L_\alpha$.
Notice that $\SG_4$ acts on the set $S$ by taking
\begin{eqnarray*}
\sigma:ij|kl\mapsto \sigma(i)\sigma(j)|\sigma(k)\sigma(l), 
\end{eqnarray*}
where $i,j,k,l$ are all different.
As $S$ is obviously invariant under this action of $\SG_4$, we conclude that the same is true for the set $\{L_\alpha,L_\beta,L_\gamma\}$ under the bijection (\ref{eq:isobik3st}). 
\end{proof}
This result immediately extends linearly to show that $\mathfrak{L}_{K3ST}=\lrc{L_\alpha,L_\beta,L_\gamma}$ forms a module of $\SG_4$ with dimension three, but how does this module fit into the decomposition (\ref{eq:schurdecompS4})?
Note that
\beqn
L_\alpha+L_\beta+L_\gamma=L_{id},\nonumber
\eqn
so $\mathfrak{L}_{K3ST}$ contains the trivial module $\lrc{L_{id}}\cong\{4\}$.
Referring to our decomposition (\ref{eq:schurdecompS4}), it is enough to do a dimension count to see that the only possible decomposition is
\beqn\label{eq:K3STdecomp}
\mathfrak{L}_{K3ST}\cong \{4\}\oplus \{2^2\}.
\eqn
Thus the Kimura 3ST model accounts for the third term in (\ref{eq:decompGMM}) and we are left with accounting for the module $\{21^2\}$.

To this end we define the six antisymmetric combinations
\beqn
A_{ij}:=L_{ij}-L_{ji}, \quad 1\leq i<j\leq 4.\nonumber
\eqn
Referring to Table~\ref{tab:chartabS4}, the projector onto the $\{21^2\}$ subspace is given by
\beqn
\Theta_{\s{21^2}}&=\fra{1}{\left|\SG_4\right|}\sum_{\sigma\in\SG_4}\chi^{\s{21^2}}(\sigma)\sigma\nonumber\\
&=\fra{1}{24}\left(3e-(12)-(13)-(14)-(23)-(24)-(34)-(12)(34)-(13)(24)\right.\\
&\hspace{6em}\left.-(14)(23)+(1234)+(1243)+(1324)+(1342)+(1423)+(1432)\right).
\eqn
Applying this projector to our antisymmetric combinations, we get, for example,
\beqn\label{eq:211proj}
\Theta_{\s{21^2}}\cdot A_{12}=\fra{1}{12}\left(2A_{12}-A_{13}-A_{14}+A_{23}+A_{24}\right),
\eqn
from which it is easy to see that the general rule is
\beqn
P_{ij}:=\fra{1}{12}\Theta_{\{21^2\}}\cdot A_{ij}=2A_{ij}-A_{ik}-A_{il}+A_{jk}+A_{jl},\nonumber
\eqn
where $i,j,k,l$ are all different.
We can also see from this that we have (at least) three linearly independent relations:
\beqn
P_{12}+P_{13}+P_{14}=P_{12}+P_{23}+P_{24}=P_{13}+P_{23}+P_{34}=0,\nonumber
\eqn
so the dimension of $\lrc{\{P_{ij}\}_{1\leq i<j\leq 4}}$ is at most three.
However, as the projection (\ref{eq:211proj}) is non-zero, and $\{21^2\}$ is three dimensional, it must be the case that
\beqn
\langle \{P_{ij}\}_{1\leq i<j\leq 4}\rangle_{\mathbb{C}}\cong\{21^2\}.\nonumber
\eqn

Putting all of these results together:
\begin{res}\label{eq:schurbasis}
The Lie algebra $\mathfrak{L}_{GMM}$ of the four state general Markov model can be expressed as
\beqn
\mathfrak{L}_{GMM}&=\lrc{\{L_{ij}\}_{1\leq i\neq j\leq 4}}\\
&\cong \lrc{\{L_{id}\}\cup \{L_\alpha,L_\beta,L_\gamma\}\cup\{R_{i}\}_{1\leq i\leq 4}\cup\{C_{i}\}_{1\leq i\leq 4}\cup\{P_{ij}\}_{1\leq i<j \leq 4}},\nonumber
\eqn
with linear dependencies
\beqn
L_{id}&=L_\alpha+L_\beta+L_\gamma=R_1+R_2+R_3+R_4=C_1+C_2+C_3+C_4,\\
P_{12}+P_{13}+P_{14}&=P_{12}+P_{23}+P_{24}=P_{13}+P_{23}+P_{34}=0,\nonumber
\eqn
and decomposition into irreducible representations of $\SG_4$:
\beqn
\lrc{L_{id}}&\cong \{4\},\nonumber\\
\lrc{L_\alpha,L_\beta,L_\gamma}&\cong \{4\}\oplus\{2^2\},\\
\lrc{\{R_{i}\}_{1\leq i\leq 4}}&\cong \lrc{\{C_{i}\}_{1\leq i\leq 4}}\cong  \{4\}\oplus \{31\},\\
\lrc{\{P_{ij}\}_{1\leq i<j\leq 4}}&\cong \{21^2\}.\nonumber
\eqn
\end{res}

We would of course like to understand the Lie algebra of $\mathfrak{L}_{GMM}$ in this basis.
As discussed at the beginning of this section, we are unfortunately bereft of theoretical results that would take us directly from (\ref{eq:GMMalg}) and Result~\ref{eq:schurbasis} to give the Lie algebra in this basis.
Instead, resorting to tedious matrix computations we have found:
\begin{res}
Using the basis defined by the decomposition into irreducible representations of $\SG_4$ given in Result~\ref{eq:schurbasis}, the Lie algebra of the general Markov model (\ref{eq:GMMalg}) can be expressed as: 
\beqn\label{eq:algebrasS4}
\left[R_i,R_j\right]&=R_i-R_j,\\
\left[L_\alpha,L_\beta\right]&=\lie{L_\alpha}{L_\gamma}=\lie{L_\beta}{L_\gamma}=0,\\
\left[L_{ij|kl},R_{i}\right]&=R_j-R_i,\\
\left[C_i,C_j\right]&=R_j-R_i-P_{ij},\\
\left[C_i,R_j\right]&=\delta_{ij}\left(L_{id}-R_{j}\right),\\
\left[C_{i},L_{ij|kl}\right]&=R_j-R_i-P_{ij},
\eqn
with the understanding that $P_{ij}=-P_{ji}$, $L_{12|34}=L_{\alpha}$, $L_{13|24}=L_{\beta}$ and $L_{14|23}=L_\gamma$.
\end{res}
This nicely gives us an alternative presentation of the Lie algebra of the general Markov model (\ref{eq:GMMalg}) in the basis that explicitly presents the decomposition into irreducible modules of $\SG_4$.

\private{Jeremy: This is the end of the ``A convenient basis'' subsection. Again, the end of the subsection needs to be announced, but I'm not sure how to do it.}
\subsection{Application of the general method.}

Following the general scheme of Section 5, our task now is to identify Lie Markov models occurring as sub-algebras in (\ref{eq:algebrasS4}).
In Table~\ref{tab:S4models} we present the decomposition of the orbits of $\SG_4$.
These are computed by using the orbit stabilizer theorem and projecting $\lrc{\SG_4/H}$ onto the irreducible module $V^\lambda$ of $\SG_4$ using the projection operator $\Theta_\lambda$.
These computations are quite tedious, so here we will simply develop the case $H=\mathbb{Z}_2\wr \mathbb{Z}_2$ as an illustrative example.

First of all we note that there are three copies of $H=\mathbb{Z}_2\wr \mathbb{Z}_2$ in $\SG_4$:
\beqn\label{z2wrz2}
\mathbb{Z}_2\wr \mathbb{Z}_2&\cong\{e,(12),(34),(12)(34),(13)(24),(14)(23),(1324),(1423)\}\\
&\cong \{e,(13),(24),(13)(24),(12)(34),(14)(23),(1234),(1432)\}\\
&\cong \{e,(14),(23),(14)(23),(13)(24),(12)(34),(1243),(1342)\}.
\eqn
Obviously each of these copies of $\mathbb{Z}_2\wr \mathbb{Z}_2$ is structurally the same\footnote{This is because each copy can be mapped to the others by conjugation with a permutation $\sigma \in \SG_4$.} and hence will result in an isomorphic action of $\SG_4$ on $\SG_4/H$.
Choosing the first copy of $\mathbb{Z}_2\wr \mathbb{Z}_2$ in (\ref{z2wrz2}), we have 
\begin{eqnarray*}
\SG_4 \, / \, \mathbb{Z}_2\wr \mathbb{Z}_2=\left\{ [e],[(13)],[(14)] \right \}
\end{eqnarray*}
where $[\sigma]$ represents the coset in $\mathfrak{S_4}/\mathbb{Z}_2\wr \mathbb{Z}_2$ containing the element $\sigma$, so for example, 
\beqn
{[e]} & =  \{e,(12),(34),(12)(34),(13)(24),(14)(23),(1324),(1423)\}, \\
{[(13)]} & =  \{(13),(123),(134),(1234),(24),(1432),(243),(142)\}, \\
{[(14)]} & =  \{(14),(124),(143),(1243),(1342),(23),(132),(234)\}. \nonumber
\eqn
These cosets inherit an action of $\SG_4$ by taking $\sigma: [\sigma']\mapsto [\sigma\sigma']$, which can be extended linearly to a representation of $\SG_4$ by taking the module 
\beqn
\lrc{\mathfrak{S_4}\, / \,\mathbb{Z}_2\wr \mathbb{Z}_2}=\lrc{[e],[(13)],[(14)]}\cong \CC^3,\nonumber
\eqn
with action defined as follows: given $\sigma\in \SG_4$,  
a vector $v\!=\!c_1 [e]+c_2[(13)]+c_3[(14)]$ is mapped to $v':=\sigma\cdot v=c_1[\sigma]+c_2[\sigma(13)]+c_3[\sigma(14)]$. 

We would like to decompose $\lrc{\mathfrak{S_4}\, / \,\mathbb{Z}_2\wr \mathbb{Z}_2}$ into irreducible modules of $\SG_4$.
This can be achieved by applying the projection operators.
For example:
\beqn
\Theta_{\{4\}}[e]&=\fra{1}{24}\sum_{\sigma\in \SG_4}\sigma\cdot [e]\\
&=\fra{1}{24}\sum_{\sigma\in \SG_4}[\sigma]\\
&=\fra{1}{24}\left(8[e]+8[(13)]+8[(14)]\right),\nonumber
\eqn
where, in the last equality, we have identified common cosets (eg. $[(12)]\!=\![e]$ etc.)
As this projection is non-zero, we conclude that $\lrc{\mathfrak{S_4}/\mathbb{Z}_2\wr \mathbb{Z}_2}$ contains the trivial representation $\{4\}$.
It is easy to check that $\Theta_{\{4\}}[(13)]=\Theta_{\{4\}}[(14)]=\fra{1}{3}\left([e]+[(13)]+[(14)]\right)$ so in fact $\lrc{\mathfrak{S_4}/H}$ contains $\{4\}$ only \emph{once}. 
Now, referring to the Table~\ref{tab:chartabS4}, we have
\beqn
\Theta_{\{31\}}[e]&=\fra{1}{24}\sum_{\sigma\in \SG_4}\chi^{\{31\}}(\sigma) \sigma\cdot [e]\\
&=\fra{1}{24}\left(3e\!+\!(12)\!+\!(13)\!+\!(14)\!+\!(23)\!+\!(24)\!+\!(34)\!-\!(12)(34)\!-\!(13)(24)\right.\\
&\hspace{4em}\left.-(14)(23)\!-\!(1234)\!-\!(1243)\!-\!(1324)\!-\!(1342)\!-\!(1423)\!-\!(1432)\right)\cdot [e]\\
&=0,\nonumber
\eqn
where again the final equality follows by identifying common cosets.
Similarly, we check that $\Theta_{\{31\}}[(13)]=\Theta_{\{31\}}[(14)]=0$ and we learn that $\lrc{\mathfrak{S_4}/\mathbb{Z}_2\wr \mathbb{Z}_2}$ does \emph{not} contain the representation $\{31\}$.
In any case we could have concluded this from the outset by noting that both $\{31\}$ and $\lrc{\mathfrak{S_4}/\mathbb{Z}_2\wr \mathbb{Z}_2}$ are three-dimensional and we have already accounted for one of these dimensions of $\lrc{\mathfrak{S_4}/H}$ with the trivial representation.
Referring again to Table~\ref{tab:chartabS4}, we see that either $\lrc{\mathfrak{S_4}/H}$ contains two copies of the sign representation $\{1^4\}$ or a single copy of $\{2^2\}$.

Consider
\beqn
\Theta_{\{2^2\}}[e]&=\fra{1}{24}\sum_{\sigma\in \SG_4}\chi^{\{2^2\}}(\sigma)\cdot [e]\\\nonumber
&=\fra{1}{24}\left(2e-(123)-(124)-(134)-(143)-(234)-(243)\right.\\
&\hspace{4em}\left.+2(12)(34)+2(13)(24)+2(14)(23)\right)\cdot[e]\\
&=\fra{1}{24}\left(2[e]\!-\![(13)]\!-\![(14)]\!-\![(13)]\!-\![(14)]\!-\![(14)]\!-\![(13)]\!+\!2[e]\!+\!2[e]\!+\!2[e]\right)\\
&=\fra{1}{24}\left(8[e]-3[(13)]-3[(14)]\right)
\eqn 
and
\beqn
\Theta_{\{2^2\}}[(13)]&=\fra{1}{24}\sum_{\sigma\in \SG_4}\chi^{\{2^2\}}(\sigma)\cdot [(13)]\\\nonumber
&=\fra{1}{24}\left(2e-(123)-(124)-(134)-(143)-(234)-(243)\right.\\
&\hspace{4em}\left.+2(12)(34)+2(13)(24)+2(14)(23)\right)\cdot[(13)]\\
&=\fra{1}{24}\left(2[(13)]-[(23)]-[(1324)]-[(14)]-[(34)]-[(1423)]\right.\\
&\hspace{4em}\left.-[(1243)]+2[(1432)]+2[(24)]+2[(1234)]\right)\\
&=\fra{1}{24}\left(8[(13)]-3[(23)]-3[e]\right).
\eqn

Putting this together we have
\beqn
\lrc{\mathfrak{S_4}/H}&=\lrc{[e]+[(13)]+[(14)],8[e]-3[(13)]-3[(14)],8[(13)]-3[(23)]-3[e]}\\
&\cong \{4\}\oplus \{2^2\}.\nonumber
\eqn

Proceeding as in this example, we have produced the results summarized in Table~\ref{tab:S4models}.
In particular, Table~\ref{tab:S4models} gives the decomposition of the action of $\SG_4$ on $\lrc{\SG_4/H}$ into irreducible representations for each subgroup $H\leq \SG_4$.
The second column records how many copies of each subgroup $H$ occur in $\SG_4$, with non-isomorphic copies accounted for with distinct decomposition in the fourth column.
For example, there are two ``types'' of $\mathbb{Z}_2$ in $\SG_4$:
\beqn
\mathbb{Z}_2\cong \{e,(12)\}\cong \{e,(13)\}\cong \{e,(14)\}\cong \{e,(23)\} \cong \{e,(24)\} \cong \{e,(34)\},\nonumber
\eqn
or 
\beqn
\mathbb{Z}_2\cong \{e,(12)(34)\}\cong \{e,(13)(24)\}\cong \{e,(14)(23)\}.\nonumber
\eqn
These two types are structurally different and as a result, the corresponding spaces $\lrc{\SG_4/H}$ have differing decomposition into irreducible subspaces, as shown in Table~\ref{tab:S4models}.

Similarly, there are two ``types'' of $\mathbb{Z}_2\times \mathbb{Z}_2$:
\beqn
\mathbb{Z}_2\times \mathbb{Z}_2\cong \{e,(12),(34),(12)(34)\}\cong \{e,(13),(24),(13)(24)\}\cong \{e,(14),(23),(14)(23)\},\nonumber
\eqn
and
\beqn
\mathbb{Z}_2\times \mathbb{Z}_2\cong \{e,(12)(34),(13)(24),(14)(23)\}.\nonumber
\eqn
Again, these two types have differing decomposition into irreducible subspaces, as shown in Table~\ref{tab:S4models}.

\begin{table}[t]
\centering
{\renewcommand{\arraystretch}{1.2}
\begin{tabular}{c|c|c|c|c}
$H\leq \SG_4$ & Copies & Cardinality$=\frac{|\SG_4|}{|H|}$ & Decomposition of $\lrc{\SG_4 / H}$ & Model \\
\hline
$\{e\}$ & 1 & 24 & $\{4\} \!\oplus\! 3\s{31}\!\oplus\! 2\s{2^2}\!\oplus\! 3\s{21^2}\!\oplus\! \{1^4\}$ & --\\
$\mathbb{Z}_2$ & 6 & 12 & $\{4\}\!\oplus\! 2\s{31}\!\oplus\! \s{2^2}\!\oplus\! \s{21^2}$ & GMM \\ 
\texttt{"} & 3 & \texttt{"} & $\{4\}\oplus \{31\}\oplus 2\{2^2\}\oplus \{21^2\}\oplus \{1\}$ & -- \\
$\mathbb{Z}_3$ & 4 & 8 & $\{4\}\!\oplus\! \s{31}\!\oplus\! \s{21^2}\!\oplus\! \{1^4\}$ & -- \\
$\mathbb{Z}_4$ & 3 & 6 & $\{4\}\!\oplus\! \s{2^2} \!\oplus\! \s{21^2}$ & -- \\
$\mathbb{Z}_2\times \mathbb{Z}_2$ & 3 & 6 & $\{4\}\!\oplus\! 2\s{2^2}\!\oplus\! \{1^4\}$  & --  \\
\texttt{"} & 1 & \texttt{"} & $\{4\}\!\oplus\! \s{31}\!\oplus\! \s{2^2}$ &  F81+K3ST \\
$\SG_3$ & 4 & 4 & $\{4\}\!\oplus\!\s{31}$ & F81 \\
$\mathbb{Z}_2 \wr \mathbb{Z}_2$ & 3 & 3 & $\{4\}\!\oplus\! \s{2^2}$ & K3ST \\
$A_4$ & 1 & 2 & $\{4\}\!\oplus\! \{1^4\}$ & -- \\
$\SG_4$ & 1 & 1 & $\{4\}$ & Jukes-Cantor 
\end{tabular}}
\caption{Decomposition of the orbits of $\SG_4$ into irreducible modules.}
\label{tab:S4models}
\end{table}

The final column in Table~\ref{tab:S4models} gives the name of the Lie Markov model that has the $\SG_4$ symmetry defined by $\SG_4/H$.
Presently we will show how these Lie Markov models arise by following the general scheme outlined in \S\ref{sec4}.

First we note that both the decomposition (\ref{eq:schurdecompS4}) of $\mathfrak{L}_{GMM}$ and the decomposition of each $\lrc{\SG_4/H}$ in Table~\ref{tab:S4models} have exactly one copy of the trivial representation, hence we observe that any Lie Markov model must occur as a \emph{single} orbit under the action of $\SG_4$ and as a consequence:
\begin{res}
In the four state case, there are no Lie Markov models with $\SG_4$ symmetry with dimension five, seven, nine, ten or eleven.
\end{res}


\subsubsection*{Dimension One}

From Table~\ref{tab:S4models} we see that there is only one abstract orbit of $\SG_4$ with cardinality one.
Thus, any orbit of cardinality one is isomorphic to $\SG_4/\SG_4$, with decomposition $\lrc{\SG_4/\SG_4}\cong \{4\}$.
As the general Markov model contains one copy of the trivial representation, we conclude:
\begin{res}
In the four state case, there is only a single one-dimensional Lie Markov model with $\SG_4$ symmetry.
\end{res}
It is immediate that the choice $\mathfrak{L}_{JC}:=\lrc{L_{id}}$ provides a stochastic basis and this model is nothing but the Jukes-Cantor model \citep{jukes1969} with a typical rate-matrix taking the form:
\beqn
Q=\alpha L_{id}=
\left(
\begin{array}{cccc}
	-3\alpha & \alpha & \alpha &\alpha  \\
	\alpha  & -3\alpha & \alpha & \alpha \\
	\alpha & \alpha & -3\alpha & \alpha \\
	\alpha & \alpha & \alpha & -3\alpha
\end{array}\right),\nonumber
\eqn
with associated graph given in Figure~\ref{fig:JCgraph}.
It is also worth noting that, because $\SG_n/\SG_n$ always has cardinality one, this result extends easily to the $n$-state case.

\begin{figure}
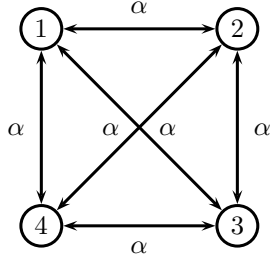

\centering
$
\psmatrix[colsep=2cm,rowsep=2cm,linewidth=.04cm,mnode=circle]
1&2\\
4&3
\ncline{<->}{1,1}{1,2}^\alpha
\ncline{<->}{2,1}{2,2}_\alpha
\ncline{<->}{1,1}{2,1}<\alpha
\ncline{<->}{1,2}{2,2}>\alpha
\ncline{<->}{1,1}{2,2}<\alpha
\ncline{<->}{1,2}{2,1}>\alpha
\endpsmatrix
$
\caption{Graphical representation of the Jukes-Cantor model.}
\label{fig:JCgraph}
\end{figure}

\subsubsection*{Dimension Two}

From Table~\ref{tab:S4models} we see that the orbit with cardinality two is $\SG_4/A_4$ with decomposition $\lrc{\SG_4/A_4}\cong \{4\}\oplus \{1^4\}$.
However the general Markov model does not contain a copy of the sign representation $\{1^4\}$, so we conclude immediately:
\begin{res}
In the four state case, there is no two-dimensional Lie Markov model with $\SG_4$ symmetry.
\end{res}

\subsubsection*{Dimension Three}

From Table~\ref{tab:S4models} we see that the only orbit with cardinality three is $\SG_4/{\mathbb{Z}_2\wr\mathbb{Z}_2}$, with decomposition $\lrc{\SG_4/{\mathbb{Z}_2\wr\mathbb{Z}_2}}\cong \{4\}\oplus \s{2^2}$.
As we saw in (\ref{eq:K3STdecomp}), this subspace is given by
\beqn
\lrc{L_{\alpha},L_{\beta},L_{\gamma}}\cong \{4\}\oplus \s{2^2},\nonumber
\eqn
with \emph{abelian} Lie algebra\footnote{We note here that the fact that the Lie algebra of the Kimura 3ST model is abelian was first explicitly discussed in \citet{bashford2004}.}
\beqn
\lie{L_\alpha}{L_\beta}=\lie{L_\alpha}{L_\gamma}=\lie{L_\beta}{L_\gamma}=0.\nonumber
\eqn
A  generic rate-matrix in this model looks like 
\begin{eqnarray*}
 Q=\left(
\begin{array}{cccc}
	* & \alpha & \beta & \gamma  \\
	\alpha  & * & \gamma & \beta \\
	\beta & \gamma & * & \alpha \\
	\gamma & \beta & \alpha & *
\end{array}\right)\tb{=\alpha L_\alpha+\beta L_\beta+\gamma L_\gamma,}
\end{eqnarray*}
where $*=-\alpha-\beta-\gamma$.
This is, of course, the Kimura 3ST model \citep{kimura1981} with associated graph given in Figure~\ref{fig:K3STgraph}.
\begin{res}
In the four state case, the only three-dimensional Lie Markov model with $\SG_4$ symmetry is the Kimura 3ST model.
\end{res}

\subsubsection*{Dimension Four}

From Table~\ref{tab:S4models} we see that the only orbit with cardinality four is $\SG_4/{\SG_3}$, with decomposition $\lrc{\SG_4/{\SG_3}}\cong \{4\}\oplus \s{31}$.
As we saw in (\ref{eq:Rdecomp}), this subspace is given by either
\beqn
\lrc{R_1,R_2,R_3,R_4}\cong \{4\}\oplus\{31\},\nonumber
\eqn
or
\beqn
\lrc{C_1,C_2,C_3,C_4}\cong \{4\}\oplus\{31\}.\nonumber
\eqn
By referring to (\ref{eq:algebrasS4}), we see that of these two possibilities only $\lrc{R_1,R_2,R_3,R_4}$ forms a Lie algebra: 
\beqn
\lie{R_i}{R_j}=R_i-R_j.\nonumber
\eqn
A generic rate-matrix in this model looks like 
\begin{eqnarray*}
 Q=\left(
\begin{array}{cccc}
	\ast & a & a & a  \\
	b  & \ast & b & b \\
	c & c & \ast & c \\
	d & d & d & \ast
\end{array}\right)\tb{=aR_1+bR_2+cR_3+dR_4,}
\end{eqnarray*}
where the diagonal entries are determined by the zero column-sum condition.
This is of course the Felsenstein 81 model \citep{felsenstein1981} with associated graph given in Figure~\ref{fig:F81graph}.
Thus we have:
\begin{res}
In the four state case, the only four-dimensional Lie Markov model with $\SG_4$ symmetry is the Felsenstein 81 model.
\end{res}

\begin{figure}
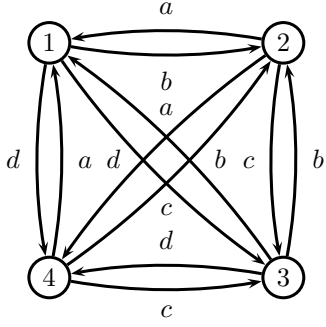

\centering
$
\psmatrix[colsep=2.5cm,rowsep=2.5cm,linewidth=.04cm,mnode=circle]
1&2\\
4&3
\ncarc[arcangle=10]{<-}{1,1}{1,2}^{a}
\ncarc[arcangle=-10]{->}{1,1}{1,2}_{b}
\ncarc[arcangle=10]{<-}{2,1}{2,2}^{d}
\ncarc[arcangle=-10]{->}{2,1}{2,2}_{c}
\ncarc[arcangle=10]{<-}{1,1}{2,1}>{a}
\ncarc[arcangle=-10]{->}{1,1}{2,1}<{d}
\ncarc[arcangle=10]{<-}{1,2}{2,2}>{b}
\ncarc[arcangle=-10]{->}{1,2}{2,2}<{c}
\ncarc[arcangle=10]{<-}{1,1}{2,2}^{a}
\ncarc[arcangle=-10]{->}{1,1}{2,2}_{c}
\ncarc[arcangle=10]{<-}{1,2}{2,1}>{b}
\ncarc[arcangle=-10]{->}{1,2}{2,1}<{d}
\endpsmatrix
$
\caption{Graphical representation of the Felsenstein 81 model.}
\label{fig:F81graph}
\end{figure}

\subsubsection*{Dimension Six}

At cardinality six, we have the orbit $\SG_4/{\mathbb{Z}_2\times \mathbb{Z}_2}$, but there are two non-isomorphic copies of $\mathbb{Z}_2\times \mathbb{Z}_2$ that we need to consider:
\beqn
\{e,(12),(34),(12)(34)\}\cong \{e,(13),(24),(13)(24)\}\cong \{e,(14),(23),(14)(23)\}&\cong {\mathbb{Z}_2\times \mathbb{Z}_2},\nonumber
\eqn
or
\beqn
\{e,(12)(34),(13)(24),(14)(23)\}&\cong {\mathbb{Z}_2\times \mathbb{Z}_2}\nonumber.
\eqn

In the first case, we have the decomposition \[\lrc{\SG_4/\{e,(12),(34),(12)(34)\}}\cong \{4\}\oplus 2\{21^2\}\oplus \{1^4\},\] and we see immediately that this cannot support a submodel because of the occurrence of the sign representation $\{1^4\}$ .
In the second case, we have the decomposition \[\lrc{\SG_4/\{e,(12)(34),(13)(24),(14)(23)\}}\cong \{4\}\oplus\s{31}\oplus\s{2^2},\] and we see that this orbit may support a submodel.
Referring to Result~\ref{eq:schurbasis}, we have 
\beqn
\lrc{L_\alpha,L_\beta,L_\gamma,R_1,R_2,R_3,R_4}\cong \lrc{L_\alpha,L_\beta,L_\gamma,C_1,C_2,C_3,C_4}\cong\{4\}\oplus\s{31}\oplus\s{2^2},\nonumber
\eqn
with the linear relations $L_\alpha+L_\beta+L_\gamma=R_1+R_2+R_3+R_4=C_1+C_2+C_3+C_4=L_{id}$.
However, referring to (\ref{eq:algebrasS4}), it is clear that $\lrc{L_\alpha,L_\beta,L_\gamma,C_1,C_2,C_3,C_4}$ does not form a Lie algebra and can thus be excluded.  
On the other hand, we see that $\lrc{L_\alpha,L_\beta,L_\gamma,R_1,R_2,R_3,R_4}$ forms a Lie algebra as we have the ``cross brackets'':
\beqn
\lie{L_{ij|kl}}{R_i}=R_{j}-R_i.\nonumber
\eqn
We refer to this model as $K3ST+F81$, where because of the linear relation we have 
\beqn
\dim\left(\lrc{L_\alpha,L_\beta,L_\gamma,R_1,R_2,R_3,R_4}\right)&=\dim\left(\lrc{L_\alpha,L_\beta,L_\gamma}\right)+\dim\left(\lrc{R_1,R_2,R_3,R_4}\right)-1\\
&=3+4-1=6.\nonumber
\eqn

The other possibility for dimension six is the orbit given by $\SG_4/\mathbb{Z}_4$, with decomposition $\lrc{\SG_4/\mathbb{Z}_4}\cong \{4\}\oplus \s{2^2}\oplus \s{21^2}$.
Again referring to Result~\ref{eq:schurbasis}, we have
\beqn
\lrc{\{P_{ij}\}_{i<j},L_\alpha,L_{\beta},L_\gamma}\cong \{4\}\oplus \s{2^2} \oplus\s{21^2}.\nonumber
\eqn
However this module does not form a Lie algebra, as, for example,
\beqn
\lie{K_\alpha}{P_{12}}=2 L_\alpha +2 L_\beta +2 L_\gamma-4 R_2-2 R_3-2 R_4+2 C_1-2 C_2.\nonumber
\eqn

\begin{res}
The only six-dimensional Lie Markov model with $\SG_4$ symmetry is the vector space sum of the Kimura 3ST and Felsenstein 81 models: $\lrc{L_\alpha,L_\beta,L_\gamma,R_1,R_2,R_3,R_4}$.
\end{res}

\tb{Constructing a basis for this six-dimensional model that exhibits the permutation symmetry is non-trivial, but can be achieved as follows.
Notice that the set of pairs:
\beqn
T:=\{\{1,2\},\{1,3\},\{1,4\},\{2,3\},\{2,4\},\{3,4\}\},\nonumber
\eqn
forms a set of cardinality six that is invariant under the permutations $\sigma\in \SG_4$.
Now consider the following (surjective) map between this set of pairs and the set of bipartions $S=\{12|34,13|24,14|23\}$: 
\beqn
\{1,2\}&\mapsto 12|34,\nonumber\\
\{1,3\}&\mapsto 13|24,\\
\{1,4\}&\mapsto 14|23,\\
\{2,3\}&\mapsto 14|23,\\
\{2,4\}&\mapsto 13|24,\\
\{3,4\}&\mapsto 12|34.\\
\eqn
The existence of this map motivates the construction of the tangent vectors
\beqn
W_{12}&:=L_\alpha+(R_1+R_2),\\
W_{13}&:=L_\beta+(R_1+R_3),\\
W_{14}&:=L_\gamma+(R_1+R_4),\\
W_{23}&:=L_\gamma+(R_2+R_3),\\
W_{24}&:=L_\beta+(R_2+R_4),\\
W_{34}&:=L_\alpha+(R_3+R_4),\nonumber
\eqn
and we see that we may then take $B_{F81+K3ST}=\left\{W_{12},W_{34},W_{13},W_{24},W_{14},W_{23}\right\}$ as a stochastic basis for $\LL_{F81+K3ST}$.}
\tb{A generic rate-matrix for this model has the form
\beqn
 Q:&=\alpha W_{12}+\bar{\alpha}W_{34}+\beta W_{13}+\bar{\beta}W_{24}+\gamma W_{14}+\bar{\gamma}W_{23}\\
 &=\left(
\begin{array}{cccc}
	\ast & 2\alpha+\bar{\alpha}+\beta+\gamma & 2\beta+\alpha+\bar{\beta}+\gamma & 2\gamma+\alpha+\beta+\bar{\gamma} \\
	2\alpha+\bar{\alpha}+\bar{\beta}+\bar{\gamma} & \ast & 2\bar{\gamma}+\alpha+\bar{\beta}+\gamma & 2\bar{\beta}+\alpha+\beta+\bar{\gamma} \\
	2\beta+\bar{\alpha}+\bar{\beta}+\bar{\gamma} & 2\bar{\gamma}+\bar{\alpha}+\beta+\gamma & \ast & 2\bar{\alpha}+\alpha+\beta+\bar{\gamma} \\
	2\gamma+\bar{\alpha}+\bar{\beta}+\bar{\gamma} & 2\bar{\beta}+\bar{\alpha}+\beta+\gamma & 2\bar{\alpha}+\alpha+\bar{\beta}+\gamma & \ast
\end{array}\right),
\nonumber
\eqn
where the diagonal entries are determined by the zero column-sum condition.
}

\tb{The Lie algebra of this model in the above basis can be found via explicit computation. 
If $ij|kl$ is a bipartion, then
\beqn
\lie{W_{ij}}{W_{kl}}=2\left(W_{ij}-W_{kl}\right).\nonumber
\eqn
Otherwise, if $ij|i'j'$ and $kl|k'l'$ are distinct bipartions, then
\beqn
\lie{W_{ij}}{W_{kl}}=2\left(W_{ij}-W_{i'j'}\right)-2\left(W_{kl}-W_{k'l'}\right).\nonumber
\eqn
}

\private{Jeremy: Here we need to discuss a suitable stochastic basis for this model. Is there an obvious one?}
\private{Jeremy (LATER): Ok, I found it.}

\subsubsection*{Dimension Eight}
 
From Table~\ref{tab:S4models} we see that the only $\SG_4$ orbit with cardinality eight is $\SG_4/{\mathbb{Z}_3}$, with decomposition $\lrc{\SG_4/{\mathbb{Z}_3}}\cong \{4\} \oplus \s{31} \oplus \s{21^2}\oplus \{1^4\}$.
Again, due to the occurrence of the sign representation $\{1^4\}$, we conclude that:
\begin{res}
There is no eight-dimensional Lie Markov model with $\SG_4$ symmetry.
\end{res}

Summarizing these results is the main outcome of this article:

\begin{thm}\label{thm:bigone}
On four character states, there are exactly five Lie Markov models with $\SG_4$ symmetry.
These models have dimension one, three, four, six and twelve, and Lie algebras as given in Table~\ref{tab:final}.
\end{thm}

\begin{table}\label{tab:final}
\begin{tabular}{ccc}
\hline
{Model} & {Dimension} & {Lie algebra} \\
\hline
\vspace{.2em}
GMM & 12 & $\left[L_{ij},L_{kl}\right]=\left(L_{il}-L_{jl}\right)\left(\delta_{jk}-\delta_{jl}\right)-\left(L_{kj}-L_{lj}\right)\left(\delta_{il}-\delta_{jl}\right)$ \\
\vspace{.1em}
K3ST+F81 & 6 & $\left[L_{ij|kl},R_i\right]=R_j-R_i$\\
\vspace{.2em}
Felsenstein 81 & 4 & $\lie{R_i}{R_j}=R_i-R_j$ \\
\vspace{.2em}
Kimura 3ST & 3 & $\lie{L_\alpha}{L_\beta}=\lie{L_\alpha}{L_\gamma}=\lie{L_\beta}{L_\gamma}=0$\\
\vspace{.2em}
Jukes Cantor  & 1 & $\emptyset$
\label{tab:S4hierarchy}
\end{tabular}
\caption{The complete list of four state Lie Markov models with $\SG_4$ symmetry. \tb{Note that $L_{12|34}:=L_\alpha$, $L_{13|24}:=L_\beta$ and $L_{14|23}:=L_\gamma$}.}
\end{table}

\section{Discussion}\label{conc}
In this article we have discussed closure of continuous-time Markov chains, and we have shown that requiring that the rate-matrices drawn from a model form a Lie algebra provides a sufficient condition for closure.
In \S\ref{sec2} we showed that the GTR model (which is the basis for most current phylogenetic studies) does \emph{not} satisfy the closure condition and therefore is not a Lie Markov model (see Definition \ref{stochasticbasis}).
We showed that the general Markov model is a Lie Markov model and described its corresponding Lie algebra.
In \S\ref{sec3} we gave a complete description of the two-state Lie Markov models; showing that there are actually an infinite continuum of Lie Markov models in this case. 
In \S\ref{sec4} we gave a new charactization of the concept of the symmetry of Markov models, and we went on to use this characterization to assist in the search for Lie Markov models on four character states.

In \S\ref{sec5} we outlined a general scheme that produces the complete list of Lie Markov models with a given symmetry.
The main result of the article is then Theorem~\ref{thm:bigone}, which shows that, for four character states, there are exactly five Lie Markov models with maximal symmetry, $\SG_4$.
Four of these are well known models: the Jukes-Cantor, the Kimura 3ST, the Felsenstein 81 and the general Markov model, and the fifth can be interpreted as the merging of the Kimura 3ST and Felsenstein 81 models.

The immediate avenue for future work is to explore Lie Markov models with slightly relaxed symmetry. 
As stated at the end of \S\ref{sec5}, we have made successfully applied our general methods to four-state case with relaxed symmetries $G<\SG_4$ in the cases of $G=\mathbb{Z}_2\wr \mathbb{Z}_2$, $\mathbb{Z}_2\times \mathbb{Z}_2$ and $\mathbb{Z}_4$. 
The $G=\mathbb{Z}_2\wr \mathbb{Z}_2$ is of particular interest as the copy, 
\[\mathbb{Z}_2\wr \mathbb{Z}_2\cong \{e,(AG),(CT),(AG)(CT),(AC)(GT),(AT)(CG),(ACGT),(ATGC)\},\] 
is exactly the set of nucleotide permutations that preserves partitioning into purines and pyrimidines: $AG|CT:=\{\{A,G\},\{C,T\}\}$.
For example
\beqn
(AG)\cdot AG|CT&=GA|CT=AG|CT,\\\nonumber
(ACGT)\cdot AG|CT&=CT|GA=AG|CT.
\eqn
We defer presentation of the list Lie Markov models with these relaxed symmetries to a forthcoming publication.

Of course, if one were to demand that the symmetry of Lie Markov models was the trivial group $G\!=\!\{e\}$, this would amount to taking no particular symmetry at all.
At this point our methods break down and one is simply back at asking for all the sub-algebras of $\mathfrak{L}_{GMM}$ that have a stochastic basis.
Unfortunately, as was discussed at the start of \S~\ref{sec5}, the methods outlined in the paper cannot be used to address the question of Lie Markov models at this level of generality.

Our presentation of Lie Markov models for a given symmetry has very desirable properties in terms of model selection. 
For instance, the practicing biologist may wish that candidate models do not provide any natural groupings of nucleotides, and hence the $\SG_4$ symmetry is appropriate and our list of five models are the appropriate models and it is then a matter of choosing how many free parameters are appropriate for the given data set.
On the other hand, the biologist may wish to distinguish between purines and pyrimidines.
In this case, as discussed above, the (forthcoming) hierarchy corresponding to Lie Markov models with $\mathbb{Z}_2\wr \mathbb{Z}_2$ symmetry and ordered by number of free parameters would be most appropriate.

Another avenue of interesting theoretical research is to explore expanding the definition of a Lie Markov model arising from a Lie algebra $\mathfrak{L}$ from the set $e^\mathfrak{L}$, to the Lie group of transition matrices whose tangent space at an arbitrary point is the Lie algebra $\mathfrak{L}$ (whose individual members need not arise as the exponential of a rate-matrix).
Again powerful methods of Lie theory are relevant here and simple topological questions such as ``Is $e^\mathfrak{L}$ equal to the connected component to the identity?'' are most natural.
A complete analysis would seek to give an understanding of the geometric aspects of these Lie groups. 
For example, the points where the determinant of the transition matrices is equal to zero will form an algebraic variety that acts a boundary to the group.
This justifies our comment at the start of \S\ref{sec2} that the key to understanding Markov matrices is to first understand them as Lie groups first.
This would also go a long way in clarifying the connections between the discrete Markov chains (or ``algebraic models'' \citep{pachter2005}) and the continuous-time formulation provided by Lie Markov models.

Finally, we can see from the results in this article that the success of Hadamard approach to the Kimura 3ST model \citep{hendy1989}, and generalization via Fourier analysis to arbitrary (abelian) ``group-based'' models \citep{szekely1993}, results directly from the fact that the Lie algebras of these models are abelian.
In Lie theory language this means that any representation of the K3ST Lie algebra can be decomposed into one-dimensional irreducible representations, ie. it is fully ``diagonalizable''.
It is then interesting to try to understand the other Lie Markov models by applying techniques from Lie theory such as the sequence of derived sub-algebras and Cartan sub-algebras \citep{erdmann2006}.
This point of view has significant potential to generalize the ``thin flattenings'' of \citet{casanellas2010} applied to equivariant models, although the exact connections of these two points of view are not apparent to the authors at this stage.

\subsubsection*{Acknowledgements}
This research was conducted with support from the University of Tasmania Visiting Fellows Program.
The first author is funded by the Australian Research Council Discovery Project DP0770991, and the second author is partially supported by Ministerio de Educación y Ciencia  MTM2009-14163-C02-02, and Generalitat de Catalunya, 2009 SGR 1284 (Spain).

\bibliographystyle{jtbnew}
\bibliography{masterAB}

\end{document}